\pdfoutput=1
\documentclass[11pt,letterpaper,margin=1in]{article}
\usepackage{graphicx}
\usepackage{placeins}
\usepackage{float}
\usepackage[T1]{fontenc}
\usepackage{fullpage}
\usepackage{amsmath, amssymb, amsthm, bbm}
\usepackage{thmtools}
\usepackage{mdframed}
\usepackage{mathtools}
\usepackage[dvipsnames]{xcolor}
\usepackage[shortlabels]{enumitem}
\usepackage{hyperref}
\usepackage{stmaryrd}
\usepackage{tikz}
\usepackage{subfig}
\usepackage{nicefrac}
\hypersetup{breaklinks, urlcolor=blue, colorlinks, citecolor=green!50!black, linkcolor=blue}
\usepackage[capitalize,noabbrev,nameinlink]{cleveref}
\usepackage[ruled,vlined,linesnumbered]{algorithm2e}
\usepackage{todonotes}
\usetikzlibrary{decorations.pathreplacing, arrows.meta}
\usepackage{makecell}

\Crefname{algocf}{Algorithm}{Algorithms}
\Crefname{claim}{Claim}{Claims} 
\Crefname{problem}{Problem}{Problems}
\Crefname{fact}{Fact}{Facts}
\Crefname{observation}{Observation}{Observations}

\newcommand{\bs}{\boldsymbol{s}}
\newcommand{\bx}{\boldsymbol{x}}
\newcommand{\bv}{\boldsymbol{v}}
\newcommand{\bw}{\boldsymbol{w}}
\newcommand{\by}{\boldsymbol{y}}

\newcommand{\bz}{\boldsymbol{z}}

\newcommand{\eps}{\varepsilon}

\newcommand{\A}{\mathcal{A}}
\newcommand{\B}{\mathcal{B}}
\newcommand{\E}{\mathbb{E}}

\newcommand{\I}{\mathcal{I}}

\newcommand{\M}{\mathcal{M}}

\renewcommand{\O}{\mathcal{O}}
\renewcommand{\P}{\mathcal{P}}
\newcommand{\Q}{\mathcal{Q}}
\newcommand{\R}{\mathbb{R}}

\newcommand{\U}{\mathcal{U}}
\newcommand{\V}{\mathcal{V}}

\newcommand{\otilde}{\widetilde{O}}

\newcommand{\defeq}{\stackrel{\mathrm{{\scriptscriptstyle def}}}{=}}

\newcommand{\Vin}{V_{\mathrm{in}}}
\newcommand{\Vout}{V_{\mathrm{out}}}
\newcommand{\Vunfound}{V_{\mathrm{unfound}}}
\newcommand{\Vactive}{V_{\mathrm{active}}}
\newcommand{\Vremoved}{V_{\mathrm{removed}}}
\newcommand{\delOf}[1]{#1_{\mathrm{del}}}

\newcommand{\Edel}{\delOf{E}}
\newcommand{\Fdel}{\delOf{F}}
\newcommand{\Mdel}{\delOf{M}}
\newcommand{\tightOf}[1]{#1_{\mathrm{tight}}}
\newcommand{\Gtight}{\tightOf{G}}
\newcommand{\Etight}{\tightOf{E}}

\DeclareMathOperator{\poly}{poly}

\declaretheorem[numberwithin=section,refname={Theorem,Theorems},Refname={Theorem,Theorems}]{theorem}
\declaretheorem[numberlike=theorem]{lemma}
\declaretheorem[numberlike=theorem]{property}

\declaretheorem[numberlike=theorem]{corollary}

\declaretheorem[numberlike=theorem,style=definition]{definition}

\declaretheorem[numbered=no,style=remark]{remark}
\declaretheorem[numberlike=theorem]{observation}

\newmdtheoremenv{theo}{Theorem}

\title{From Unweighted to Weighted Dynamic Matching in Non-Bipartite Graphs: A Low-Loss Reduction}
\author{
Aaron Bernstein\thanks{New York University, \texttt{aaron.bernstein@nyu.edu}. Supported by Sloan Fellowship, Google Research Fellowship,  NSF Career Grant 1942010, and Charles S. Baylis endowment at NYU.
    } \and
Jiale Chen\thanks{
  Stanford University,
  \texttt{jialec@stanford.edu}. 
  Supported in part by a Stanford MS\&E departmental fellowship, a Microsoft Research Faculty Fellowship, and NSF CAREER Award CCF-1844855.
}
}
\date{}

\begin{document}

\begin{titlepage}
  \maketitle \pagenumbering{roman}
  \setcounter{tocdepth}{3}
  \begin{abstract}

We study the approximate maximum weight matching (MWM) problem in a fully dynamic graph subject to edge insertions and deletions. We design meta-algorithms that reduce the problem to the unweighted approximate maximum cardinality matching (MCM) problem. Despite recent progress on \emph{bipartite} graphs -- Bernstein-Dudeja-Langley (STOC 2021) and Bernstein-Chen-Dudeja-Langley-Sidford-Tu (SODA 2025) -- the only previous meta-algorithm that applied to \emph{non-bipartite} graphs suffered a $\frac{1}{2}$ approximation loss (Stubbs-Williams, ITCS 2017). We significantly close the weighted-and-unweighted gap by showing the first low-loss reduction that transforms any fully dynamic $(1-\eps)$-approximate MCM algorithm on bipartite graphs into a fully dynamic $(1-\eps)$–approximate MWM algorithm on general (not necessarily bipartite) graphs, with only a $\poly(\log n/\eps)$ overhead in the update time. Central to our approach is a new primal–dual framework that reduces the computation of an approximate MWM in general graphs to a sequence of approximate induced matching queries on an auxiliary bipartite extension. In addition, we give the first conditional lower bound on approximate partially dynamic matching with worst-case update time.

\end{abstract}

  \newpage
  \tableofcontents
  \newpage
\end{titlepage}
\newpage
\pagenumbering{arabic}

\section{Introduction}
This paper studies the dynamic matching problem, where a graph $G$ undergoes a sequence of edge insertions and deletions (called updates), and the goal is to maintain a large matching in $G$ at all times while minimizing the update time. There are conditional bounds ruling out efficient algorithms for exact maximum matching (see, e.g., \cite{HenzingerKNS15}), so the focus is typically on approximate maximum matching. This is one of the most studied problems in dynamic algorithms, and there is extensive literature considering different approximation/update time tradeoffs for the problem, as well as other features such as whether the algorithm is amortized or worst-case, robust or oblivious. 

Almost all of the existing literature is restricted to maximum \emph{cardinality} matching (MCM), where every edge has the same weight, and the maximum matching is simply the one with the largest number of edges. By contrast, there are relatively few dynamic results on the more general maximum \emph{weight} matching problem (MWM), where every edge has a weight and the maximum matching is the one with the largest total edge weight. 

To narrow this gap, there has been some recent work on meta-algorithms that can convert any dynamic MCM algorithm into a dynamic MWM algorithm. As far as we know, the first such result is by Stubbs and Williams \cite{StubbsW17}. The result works for almost all variants of the problem and incurs almost no overhead in the update time, but it incurs a multiplicative loss of $\frac{1}{2}$ in the approximation ratio, which is a significant loss for the matching problem, where the MCM approximations are typically $(1-\eps)$ or a constant $\geq \frac{1}{2}$. Bernstein, Dudeja, and Langley then showed that in \emph{bipartite} graphs, there is a meta-algorithm that only incurs a $(1-\eps)$ loss in the approximation, while the overhead is $O_{\eps}(\log n)$ in the update time, although the dependence on $\eps$ is exponential \cite{BernsteinDL21}. Finally, a very recent result of Bernstein, Chen, Dudeja, Langley, Sidford and Tu \cite{BernsteinCDLST24} further reduced the overhead to $\poly(1/\eps)$, although the loss of their algorithm in the approximation ratio is small only when the original MCM algorithm itself is a $(1-\eps)$-approximation with small $\eps$ (more details at the end of~\cref{sec:our-results}).

The low-loss conversions of \cite{BernsteinCDLST24} and \cite{BernsteinDL21} are both limited to bipartite graphs. For non-bipartite graphs, some individual MWM results are known, but the only general conversion is that of \cite{StubbsW17}, which loses a factor of $\frac{1}{2}$ in the approximation ratio.

Meta-algorithms with $(1-\eps)$ loss in the approximation are known in other models. Gamlath, Kale, Mitrovic, and Svensson \cite{GamlathKMS19} designed a meta-algorithm in the semi-streaming and MPC models through generalizing McGregor's unweighted algorithm \cite{McGregor05} to find weighted augmenting paths and cycles; thus, similar to  McGregor's algorithm, theirs had an exponential dependency on $\eps$. Huang and Su \cite{HuangS23} instead generalized the sequential primal-dual algorithm of~\cite{DuanP14} to parallel, semi-streaming, and distributed models with polynomial dependency on $\eps$.

\subsection{Our Results: Fully Dynamic}\label{sec:our-results}
Our main result is a low-loss meta-algorithm for non-bipartite graph in the dynamic setting: the loss in the approximation ratio is from $(1-\poly(\eps))$ to $(1-\eps)$
and the overhead for the update time is $\poly(\log n/\eps)$. 

\begin{theorem}[Informal version of~\cref{thm:reduction:fully dynamic}]
\label{thm:main}
  Given a $(1-\eps)$-approximate fully dynamic MCM algorithm that, on an input $n$-vertex $m$-edge \textbf{bipartite} graph, has initialization time $\I(n,m,\eps)$, amortized update time $\U(n,m,\eps)$. Then, there is a randomized $(1-\eps)$-approximate fully dynamic MWM algorithm (w.h.p.\ against an adaptive adversary), on an input $n$-vertex $m$-edge $W$-aspect ratio \textbf{general} graph (not necessarily bipartite), that has initialization time
    \[\poly(\log n/\eps)\cdot (\I(O(n/(\eps\log(\eps^{-1}))),O(m/(\eps\log(\eps^{-1}))),\poly(\eps))+m),\]
    and amortized update time
    \[\poly(\log n/\eps)\cdot \U(O(n/(\eps\log(\eps^{-1}))),O(m/(\eps\log(\eps^{-1}))),\poly(\eps)).\]
\end{theorem}

Our result makes significant progress towards closing the weighted-and-unweighted gap, as it is the first low-loss meta-algorithm for dynamic non-bipartite graphs, and up to $\poly(\log n/\eps)$ factors in the update time, it matches the meta-algorithm for bipartite graphs of \cite{BernsteinCDLST24}. Combined with existing results for $(1-\eps)$-approximate MCM, the theorem implies several new results for $(1-\eps)$-approximate MWM in non-bipartite graphs.

\paragraph{Hypothetical Polynomially Sublinear Update Time} A key implication of our Theorem \ref{thm:main} is that we could extend a hypothetical MCM algorithm in bipartite graphs with polynomially sublinear update time to also work for MWM in general graphs. One of the most important open problems in dynamic (unweighted) matching is whether or not it is possible to obtain a fully dynamic algorithm for $(1-\eps)$-approximate MCM in polynomially sublinear update time (i.e., $n^{1-c}$ or even $n^{1-f(\eps)}$). On the hardness side, there are conditional lower bounds ruling this out for extremely small $\eps=n^{-\Omega(1)}$, though the condition is slightly weaker than standard (\emph{approximate} OMV)~\cite{Liu24}.
On the positive side, Bhattacharya, Kiss, and Saranurak showed that such a result is indeed possible if $\eps$ is a fixed constant and the algorithm only needs to return the matching \emph{size}, not the matching itself \cite{BhattacharyaKS23dynamic1}. But the general problem remains open, and our Theorem \ref{thm:main} shows that any positive result for MCM would immediately imply a positive result for MWM.

\paragraph{Slightly Sublinear Update Time}
In line with the above open problem, a very recent result of Liu \cite{Liu24} showed that a \emph{slightly} sublinear update time of $n/2^{\sqrt{\log n}}$ is in fact possible for fully dynamic $(1-\eps)$-approximate MCM in a bipartite graph. Their algorithm is later generalized to MCM in general graphs by Mitrovic and Sheu~\cite{MitrovicS25}.

\begin{lemma}[\cite{MitrovicS25}]
There is a randomized fully dynamic algorithm that maintains a $(1-\eps)$-approximate MCM in a general graph with update time $O(\poly(\eps^{-1})\cdot \frac{n}{2^{\Omega(\sqrt{\log n})}})$. 
\end{lemma}

Our meta-theorem immediately implies a similar result for weighted matching in a general graph.
\begin{corollary}
There is a randomized fully dynamic algorithm that maintains a $(1-\eps)$-approximate MWM in a general graph with update time $O(\poly(\eps^{-1})\cdot \frac{n}{2^{\Omega(\sqrt{\log n})}})$. 
\end{corollary}
\paragraph{Connection to Ruzsa-Szemeredi graphs}
Also relating to the above open problem, a recent result by Behnezhad and Ghafari~\cite{BehnezhadG24} showed an algorithm whose update time depends on the maximum density of an extremal object known as Ordered Ruzsa-Szemeredi (ORS) graphs -- if these graphs cannot be too dense, then the algorithm achieves the desired sublinear update time\footnote{A follow-up work by Pratt then relates ORS graphs to the more standard Ruzsa-Szemeredi graphs~\cite{Pratt25}, but to avoid a messy conversion, we state the results in terms of ORS.}. Their algorithm is later improved by Assadi, Khanna, and Kiss~\cite{AssadiKK25} and  Mitrovic and Sheu~\cite{MitrovicS25}.
\begin{lemma}[\cite{MitrovicS25}]
Given any $\eps\in(0,1/4)$, there is a randomized fully dynamic algorithm that for any constant $\beta>0$, maintains a $(1-\eps)$-approximate MCM in an $n$-vertex general graph with update time \[\poly(\eps^{-1})\cdot n^\beta\cdot \mathrm{ORS}(n,\poly(\eps)\cdot n),\] where $\mathrm{ORS}(n,r)$ denotes the largest possible number of edge-disjoint matchings $M_1,\dots,M_t$ of size $r$ in an $n$-vertex graph such that for every $i=1,2,\dots,t$, and $M_i$ is an induced matching in subgraph $M_i\cup M_{i+1}\cup \cdots\cup M_t$. The degree of $\poly(\cdot)$ depends on $\beta$.
\end{lemma}

Our meta-theorem implies that a similar algorithm exists for weighted matching as well.
\begin{corollary}
Given any $\eps\in(0,1/4)$, there is a randomized fully dynamic algorithm that for any constant $\beta>0$, maintains a $(1-\eps)$-approximate MWM in an $n$-vertex general graph with update time \[\poly(\eps^{-1})\cdot n^\beta\cdot \mathrm{ORS}(n,\poly(\eps)\cdot n),\]
where the degree of $\poly(\cdot)$ depends on $\beta$.
\end{corollary}

\paragraph{Offline Matching}
A natural relaxation of the standard fully dynamic model is to assume that the algorithm knows the sequence of edge updates in advance -- this is called the \emph{offline} dynamic model. Note that the problem remains highly non-trivial, as the algorithm needs to maintain a large matching after each update. Liu leveraged fast matrix multiplication to achieve polynomially sublinear update time for the offline fully dynamic $(1-\eps)$-approximate MCM in a bipartite graph~\cite{Liu24}. The result is also generalized to general graphs by Mitrovic and Sheu~\cite{MitrovicS25}.

\begin{theorem}[\cite{MitrovicS25}]
There is an offline fully dynamic randomized algorithm that maintains a $(1-\eps)$-approximate MCM in a general graph with update time $O(\poly(\eps^{-1})\cdot n^{0.58})$.
\end{theorem}

Note that our Theorem \ref{thm:main} is not stated for offline algorithms, and in fact cannot convert an arbitrary offline MCM algorithm out of the box. Instead, in the offline setting, our meta-algorithm requires the unweighted algorithm to maintain something slightly stronger than an MCM, which we call an induced matching oracle. Fortunately,~\cite{Liu24} already maintains such an oracle, so we are able to integrate it into our framework in a black-box manner.

\begin{restatable}{corollary}{OfflineDynamic}
There is an offline fully dynamic randomized algorithm that maintains a $(1-\eps)$-approximate MWM in a general graph with update time $O(\poly(\eps^{-1})\cdot n^{0.58})$.
\end{restatable}

\paragraph{Limitations of Our Meta-Algorithm} Our meta-algorithm suffers from the same limitation as the meta-algorithm of \cite{BernsteinCDLST24} for bipartite graphs: to achieve a small approximation loss of $(1-\eps)$, we require the original MCM algorithm to be a $(1-\eps')$-approximation for $\eps' = \poly(\eps)$. This can be seen in \Cref{thm:main}, where the update time of the weighted algorithm is $\U(\cdot,\cdot,\poly(\eps))$, rather than $\U(\cdot,\cdot,\eps)$. In other words, the meta-algorithm is only useful for MCM algorithms that already have a small approximation ratio. By contrast, the (bipartite-only) algorithm of \cite{BernsteinDL21} had a worse dependence on $\eps$ in the update time, but could convert any $\alpha$-approximation to an $\alpha(1-\eps)$-approximation.

The ideal meta-algorithm for non-bipartite graphs would similarly have small loss for any $\alpha$-approximation. Many state-of-the-art results for dynamic MCM in non-bipartite graphs return a $1/2$-approximation, a $2/3$-approximation, or some other fixed constant, and our Theorem \ref{thm:main} cannot convert these to weighted algorithms with small loss. Despite this limitation, our result can be applied to any $(1-\eps)$-approximate MCM algorithm, which is one of the most studied regimes for approximate matching.

\subsection{Our Results: Partially Dynamic}

$(1-\eps)$-approximate matching has also been studied in the \emph{partially} dynamic setting: the incremental setting, where edges are only inserted; and the decremental setting, where edges are only deleted. Unfortunately, unlike in the fully dynamic case, our meta-algorithm requires an additional condition here: we require the approximate MCM algorithm to have \emph{worst-case} update time.

\begin{theorem}[Informal version of~\cref{thm:reduction:partially dynamic}]
  Given a $(1-\eps)$-approximate incremental/decremental MCM algorithm that, on an input $n$-vertex $m$-edge bipartite graph, has initialization time $\I(n,m,\eps)$, \textbf{worst-case} update time $\U(n,m,\eps)$. Then, there is a $(1-\eps)$-approximate incremental/decremental MWM algorithm (w.h.p.\ against an adaptive adversary), on an input $n$-vertex $m$-edge $W$-aspect ratio general graph (not necessarily bipartite), that has initialization time
    \[\poly(\log n/\eps)\cdot (\I(O(n/(\eps\log(\eps^{-1}))),O(m/(\eps\log(\eps^{-1}))),\poly(\eps))+m),\]
    and amortized update time
    \[\poly(\log n/\eps)\cdot \U(O(n/(\eps\log(\eps^{-1}))),O(m/(\eps\log(\eps^{-1}))),\poly(\eps)).\]
\end{theorem}

Although there exist algorithms for both incremental MCM~\cite{Gupta14,GrandoniLSSS19,BlikstadK23} and decremental MCM~\cite{BernsteinGS20,JambulapatiJST22,AssadiBD22,ChenST23} with a near-optimal $\otilde_{\eps}(1)$ update time, all of these algorithms are amortized. The best known worst-case bound is $O(\sqrt{m}\eps^{-2})$ by~\cite{GuptaP13}, and there are no partially dynamic $(1-\eps)$-approximate MCM algorithms with worst-case update time that are more efficient than those fully dynamic algorithms. Thus, this meta-theorem has no immediate implications itself.

Our final result provides some explanation for this gap by showing that a \emph{decremental} $(1-\eps)$-approximate MCM algorithm with worst-case polynomially sublinear update time would imply a \emph{fully dynamic} $(1-\eps)$-approximate MCM algorithm with amortized polynomially sublinear update time, which as mentioned above, is one of the central open problems in the area. Previously, while hardness was known for the update time (both worst-case and amortized) of \emph{exact} matching~\cite{Dahlgaard16}, no hardness results were known for partially dynamic $(1-\eps)$-approximate matching.

\begin{restatable}{theorem}{PartialHardness}\label{thm:matching:decremental2full}
    Given an algorithm that solves the $(1-\eps)$-approximate decremental matching problem on bipartite graphs with $\poly(\eps^{-1}
    )\cdot n^{3-\delta}$ initialization time and $\poly(\eps^{-1})\cdot n^{1-\delta/2}$ worst-case update time for some constant $\delta>0$, then there is a randomized algorithm that solves the $(1-\eps)$-approximate fully dynamic matching problem on bipartite graphs with amortized update time $\poly(\eps^{-1})\cdot n^{1-c}$,
    for some constant $c>0$.
    We assume the initial graph is empty in the fully dynamic problem.
\end{restatable}

\section{High-Level Overview of Techniques}
\label{sec:overview}
In this section, we give a high-level outline of our framework for proving Theorem \ref{thm:main}, with an eye towards placing our techniques in the context of previous work. More detailed technical overviews can be found in the corresponding sections.

\paragraph{Primal-Dual Framework}
Our algorithm differs significantly from the previous meta-algorithms presented in \cite{BernsteinDL21,BernsteinCDLST24} that crucially assumed a bipartite graph. Instead, our algorithm is based on a sequential primal-dual framework for finding MWM in non-bipartite graphs, similar to Pettie's framework \cite{Pettie12} for exact computation, and Duan and Pettie's framework \cite{DuanP14} for approximate computation. The basic structure of these frameworks is as follows: start with trivial feasible primal and dual solutions, and successively improve the quality of the solutions (reducing the duality gap) over $t$ rounds, 
while always maintaining (approximate) feasibility. For those frameworks that compute approximate matchings, including ours, one can think of $t=\poly(\log W/\eps)$. Each round of improvement takes the following form: 
\begin{enumerate}
\item Define an \emph{unweighted} graph $\Gtight$ consisting of those edges whose corresponding dual constraints are (approximately) tight.
\item \textbf{Key Step}: compute information about the unweighted graph $\Gtight$.
\item Use the information to adjust the primal and dual solutions. 

\end{enumerate}
The main computational challenge is in the \textbf{Key Step}, through which the framework effectively reduces MWM to a small number of computations on unweighted graphs. The required information in our framework is slightly more relaxed than both~\cite{Pettie12,DuanP14}, but loosely speaking, all frameworks need to ensure certain properties about the augmenting paths and blossoms w.r.t.\ the current matching; crucially, computing an arbitrary approximate MCM on $\Gtight$ is not enough.

Such a framework is inherently amenable to other models of computation, such as streaming, parallel, and distributed. Consider, for example, the parallel model. The framework requires the algorithm to compute information about a series of unweighted graphs $\Gtight^1,\Gtight^2,\ldots,\Gtight^{t}$. Once an efficient parallel algorithm for the \textbf{Key Step} is designed, the primal-dual framework converts it to an approximate MWM algorithm, while only multiplying work and span by a $t=\poly(\log W/\eps)$ factor. Huang and Su \cite{HuangS23} used a similar approach to develop near-optimal algorithms for $(1-\eps)$-approximate MWM in the semi-streaming, parallel, and distributed models.

\paragraph{Challenges in the Dynamic Model}
\newcommand{\amcm}{\mathcal{A}_{\textrm{MCM}}}
\newcommand{\tmcm}{\mathcal{U}_{\textrm{MCM}}}
A unique challenge in the dynamic model is that the graphs $\Gtight^i$ could be extremely unstable. A single edge change to $G$ may entirely change the graph $\Gtight^i$ in the $i$-th round from the primal-dual framework. We don't have the time budget to capture all the changes if we treat $\Gtight^i$ as arbitrary subgraphs.

Let us consider this quantitatively. Say that we have a fully dynamic algorithm $\amcm$ for $(1-\eps)$-approximate MCM in bipartite graphs with update time $\tmcm$, and we want an MWM algorithm in non-bipartite graphs with update time $\poly(\log n/\eps)\cdot \tmcm$. Using standard techniques, one can show that maintaining an \emph{additive} approximate MWM suffices, and that the matching we maintain is thus approximately stable for $\poly(\eps)\cdot n$ updates (see~\cref{lemma:add-apx-to-multi-apx}). Thus, the goal of the algorithm design is to re-evaluate our MWM every $\poly(\eps)\cdot n$ updates within a time budget of $\poly(\log n/\eps)\cdot n\cdot \tmcm$. Given the recent progress of~\cite{Liu24} and the hope that a polynomially sublinear update time $\amcm$ exists, $\tmcm$ should be considered sublinear to $n$, thus the total time budget should be considered sublinear to $m$. 

Now if we use the above primal-dual framework, \textbf{Key Step} needs to be executed in sublinear time to $m$, and the computation is an impossible task if each $\Gtight^i$ is treated as an arbitrary subgraph. By contrast, in the parallel model for example, $O(m)$ work is affordable on every $\Gtight^i$.

We overcome this challenge by observing that all these graphs $\Gtight^i$ are defined by only $O(n)$ primal and dual variables: the current matched and blossom edges, the vertex duals, and a laminar family of the odd-set duals. Thus, they have a better structure than an arbitrary subgraph of $G$. We show that there exists a single graph $G^{\eps}$ -- a minor enlargement of $G$ -- which is almost as stable as $G$ is, and has the property that every possible $\Gtight^i$ can be obtained from $G^{\eps}$ using only $\poly(\log n/\eps)\cdot n$ vertex and edge updates. The initial algorithm $\amcm$ only handles edge updates, but we show that it can also handle the type of vertex updates\footnote{Here we are crucially using that $\amcm$ is a $(1-\eps)$-approximation. For arbitrary $\alpha$-approximation, converting from edges updates to vertex updates would lead to a large loss in the approximation ratio.}  involved in this work in a black box manner.
Thus, even though the \textbf{Key Step} of the primal-dual framework requires more information,  we can at least run $\amcm$ on $G^{\eps}$ and efficiently compute a $(1-\eps)$-approximate MCM on every $\Gtight^i$. This gives the hope to finally implementing the \textbf{Key Step} within the time budget.

\paragraph{Induced Matching Oracles}
To execute the \textbf{Key Step}, we use the existing abstraction of an induced matching oracle $\O$. This is a data structure on a static bipartite graph $G = L \cup R$, that can answer the following query: given sets $S \subseteq L$ and $T \subseteq R$, it returns an approximate matching in $G[S \cup T]$ (see \Cref{def:oracle}). With the same type of vertex updates as mentioned above, the fully dynamic matching algorithm $\amcm$ can serve as an induced matching oracle $\O$ on $G^\eps$ with query time $\poly(\log n/\eps)\cdot n \cdot \tmcm$. The total time budget thus allows $\poly(\log n/\eps)$ calls to $\O$. 

It was previously known how to use $\poly(\log n/\eps)$ calls to $\O$ to compute a $(1-\eps)$-approximate MCM by eliminating short augmenting paths while deleting a small number of edges and vertices (see~\cite{FischerMU22,MitrovicMSS25,MitrovicS25}). But in weighted graphs, the \textbf{Key Step} of the existing primal-dual frameworks requires even more information -- loosely speaking, these frameworks require at minimum computing a maximal set of augmenting paths (including long ones). We show that this additional information can also be computed using $\poly(\log n/\eps)$ oracle calls if we allow a small number of deleted vertices and edges. We then design a more relaxed primal-dual framework than  \cite{Pettie12,DuanP14}, which guarantees an approximate MWM at the end even when every computation of the \textbf{Key Step} allows for such deleted vertices and edges. Our new primal-dual framework also requires an additional operation to ensure that the blossoms of the intermediate solution always have small size; this is needed because each of the above ``edge'' deletions actually corresponds to deleting all edges in $G$ between two blossoms.

\section{Preliminaries}
\paragraph{General} For sets $S$ and $T$, we let
$S\oplus T\defeq (S\setminus T)\cup(T\setminus S)$ denote their symmetric difference.
\paragraph{Graphs and Matchings} Throughout this work, $G=(V,E)$ denotes an undirected graph and $\bw:E\to\R_{>0}$ is an edge weight function. The aspect ratio $W$ is the ratio between the largest edge weight and the smallest one. We assume that $W=\poly(n)$ and w.l.o.g., the edge weights in a graph with aspect ratio $W$ are between 1 and $W$. Denote $\B_G=(V_{\B_G}, E_{\B_G})$ as the unweighted bipartization of $G$, where $V_{\B_G}=\{u^L,u^R:u\in V\}$ and $E_{\B_G}=\{\{u^L,v^R\}:\{u,v\}\in E\}$. An undirected edge $\{u,v\}$ in $G$ has two corresponding directed arcs $(u,v)$ and $(v,u)$.
Denote $G[S]$ as the vertex-induced subgraph by $S$. A matching $M$ is a set of vertex-disjoint edges. We say that a matching $M^\prime$ extends another matching $M$ if and only if all matched vertices in $M$ are matched in $M^\prime$. For an edge set $F$, it weight is defined as $\bw(F)=\sum_{e\in F}\bw(e)$. Denote $\mu(G)$ as the maximum size of $M$ achieved by any matching $M$ on $G$, and $\mu_{\bw}(G)$ as the maximum value of $\bw(M)$ achieved by any matching $M$ on $G$. An $(\alpha,\delta)$-approximate MCM of $G$ is a matching $M$ such that $|M|\geq \alpha\cdot\mu(G)-\delta$ and an $(\alpha,\delta)$-approximate MWM of $G$ is a matching $M$ such that $\mu(G)\geq \alpha\cdot\mu_{\bw}(G)-\delta$. The following result states that we can compute a $(1-\eps,0)$-approximate MWM efficiently.

\begin{theorem}[{\cite[Theorem 3.12]{DuanP14}}]\label{thm:DP14}
    On an $m$-edge general weighted graph, a $(1-\eps,0)$-approximate MWM can be computed in time $O(m\log(\eps^{-1})\eps^{-1})$.
\end{theorem}

\paragraph{Alternating and Augmenting Paths}
An alternating path is a simple path containing a sequence of alternately matched and unmatched edges (it can start with either). The length of an alternating path is the number of edges in that path. An augmenting path is an alternating path whose two endpoints are both free vertices.

\paragraph{Matching Polytope}
The \emph{matching polytope} of $G(V,E)$ is defined as the convex hull of the indicator vectors of matchings in $G$.
Let
\begin{equation}
\P_{G} \defeq
\left\{\bx\in \R_{\geq 0}^E:\sum_{v\in e}x_e \leq 1, \forall\;v \in V\right\}.
\end{equation}
It is a standard fact that $\P_G$ is the matching polytope of a bipartite graph $G$.
When $G$ is non-bipartite, we add odd-set constraints.
Formally, we let 
\[
\V_G \defeq \left\{B \subseteq V: |B| \geq 3\;\text{and $|B|$ is odd}\right\}
\]
be the collections of \emph{odd sets} and define
\begin{equation}
\M_{G} \defeq
\P_G \cap \left\{
\begin{array}{ll}
\sum_{\{u,v\}\in G[B]}x_{\{u,v\}} \leq \left\lfloor\frac{|B|}{2}\right\rfloor, & \forall\;B \in \V_G \\
\end{array}
\right\}.
\label{eq:matching-polytope}
\end{equation}
It is known that $\M_G$ is the matching polytope of $G$~\cite{schrijver2003combinatorial}.

\paragraph{The MWM LP}
By definition of the matching polytope, the solution of the following LP is the convex hull of MWMs.
\begin{equation}\label{LP:Primal}
\begin{array}{ll}
\text{maximize}  &\sum_{e\in E}\bw(e)\cdot x_e \\
\text{subject to} &\bx\in\M_G
\end{array}
\end{equation}
The dual of (\ref{LP:Primal}) is given below, where $\by\in\R_{\geq 0}^V$ and $\bz\in\R_{\geq 0}^{\V_{G}}$ are dual variables of vertices and odd sets respectively. For convenience, we set $\bs\in\R_{\geq 0}^E$ to be $s_{\{u,v\}}=y_u+y_v+\sum_{B\in \V_G:u,v\in B}z_B$.
\begin{equation}\label{LP:Dual}
\begin{array}{ll}
\text{minimize}  &\sum_{u\in V}y_u+\sum_{B\in\V_G}z_B\cdot \left\lfloor\frac{|B|}{2}\right\rfloor\\
\text{subject to} &s_e\geq \bw(e),~\forall\;e\in E
\end{array}
\end{equation}

\paragraph{Blossoms}
A blossom is identified with an odd set $B$ and a set of blossom edges $E_B$. Any single vertex $v\in G$ corresponds to a trivial blossom $\{v\}$ with an empty edge set. Suppose there is an odd-length
sequence of vertex-disjoint blossoms $A_0,A_1,\dots,A_k$ with corresponding edge sets $E_{A_0},E_{A_1},\dots,E_{A_k}$ and ${A_i}$
are connected in a cycle by edges $e_0,e_1\dots,e_k$, where $e_i\in A_i\times A_{i+1\pmod{k+1}}$, then $B=\cup_i A_i$ is also a blossom associated with the edge set $E_B=\cup_i E_{A_i}\cup \{e_0,e_1,\dots,e_k\}$. We say a matching $M$ \emph{respects a laminar family of blossoms $\Omega$} if for each $B\in\Omega$, $M\cap E_B=\left\lfloor\frac{|B|}{2}\right\rfloor$. 
A blossom in the laminar family $\Omega$ is called a root blossom (or a maximal blossom) if and only if it is maximal. A blossom in the laminar family $\Omega$ respected by $M$ is called a free blossom if and only if it contains a free vertex.
For a vertex $u\in V$, $\Omega(u)$ refers to the maximal blossom in $\Omega$ that contains $u$. $G/\Omega=(V^\prime,E^\prime)$ is the \textit{graph contraction w.r.t.\ $\Omega$}, where $V^\prime=\{B\in\Omega:B~\text{is maximal}\}$ and $E^\prime=\{\{\Omega(u),\Omega(v)\}:\{u,v\}\in E\;\text{and}\;\Omega(u)\neq \Omega(v)\}$. Note that by definition, $G/\Omega$ is a simple graph.
We will use the following property to augment the matching.

\begin{lemma}[{\cite[Lemma 2.1]{DuanP14}}]\label{lemma:expanding-augmenting-paths}
    For a matching $M$ respecting a laminar family of blossoms $\Omega$,
    \begin{enumerate}[(1)]
        \item $M/\Omega$ is a matching in $G/\Omega$.
        \item Any augmenting path $P^\prime$ w.r.t.\ $M/\Omega$ in $G/\Omega$ extends to an augmenting path $P$ w.r.t.\ $M$ in $G$.
        \item For $P^\prime$ and $P$ mentioned in (2), $M\oplus P$ is a matching respecting $\Omega$.
    \end{enumerate}
\end{lemma}

\subsection{Additive versus Multiplicative Approximation}

\cref{lemma:add-apx-to-multi-apx} below allows us to focus on \emph{additive} approximate MWM.

\begin{lemma}\label{lemma:add-apx-to-multi-apx}
    Given a $(1,\eps\cdot W\cdot n)$-approximate dynamic MWM algorithm $\A$ that, on an input $n$-vertex $m$-edge $W$-aspect ratio graph, has initialization time $\I(n,m,W,\eps)$ and update time $\U(n,m,W,\eps)$, there is a $(1-O(\eps\log(\eps^{-1})),0)$-approximate fully dynamic MWM algorithm (w.h.p.\ against an adaptive adversary), on an input $n$-vertex $m$-edge $W$-aspect ratio graph, that has initialization time
    \[O(\log^2 n\cdot \eps^{-6}\log(\eps^{-1}))\cdot \left(\I(n,m,\Theta(\eps^{-2}),\Theta(\eps^6))+n\right)+O(m\eps^{-1}\log(\eps^{-1})),\]
    and amortize update time
    \[\poly(\log(\eps^{-1}))\cdot O(\log^2 n\cdot \eps^{-6})\cdot \U(n,m,\Theta(\eps^{-2}),\Theta(\eps^6)).\]
\end{lemma}

\begin{proof}[Proof Sketch] The widely known vertex sparsification lemma in the dynamic matching literature allows one to assume that an $n$-node graph has an MCM with size $\Omega(\eps \cdot n)$. An unweighted version of~\cref{lemma:add-apx-to-multi-apx} is thus immediate. We generalize the vertex sparsification lemma
to the weighted case, transforming a $(1,\poly(\eps\cdot W)\cdot n)$-approximate dynamic MWM algorithm to a $(1-\eps,0)$-approximate one with $\poly(W\log n/\eps)$ overhead and use~\cite{BernsteinCDLST24} to reduce $W$ to $\Theta(\eps^{-2})$. The full proof is deferred to~\cref{appendix:vertex-sparsification}. 
\end{proof}

\section{Reduction from Approximate MWM to Induced Matchings}
In this section, we state the main result in our paper, a black-box reduction from $(1,\eps\cdot W\cdot n)$-approximate MWM algorithms on general graphs to $(\alpha,\beta)$-approximate induced matching oracles on bipartite graphs. 
\begin{definition}
[$(\alpha,\beta)$-Approximate Induced Matching Oracle]
\label{def:oracle}
For $\alpha,\beta\in [0,1]$, an $(\alpha,\beta)$-approximate induced matching oracle $\O$ is a data structure initialized on an $n$-node graph $G=(V,E)$ that supports the following query: given $S\subseteq V$, it returns a matching $\O(S)\subseteq G[S]$
such that $|\O(S)|\geq \alpha\cdot \mu(G[S])-\beta\cdot n$.
\end{definition}
\begin{remark}
    We assume that the approximate induced matching oracle takes explicit vertex set inputs and outputs all matching edges; thus, its query time is at least $\Omega(n)$. Throughout the paper, any $O(n)$ term can be subsumed into $O(\U)$ when the oracle is initialized on an $n$-node graph, where $\U$ denotes the query time of the approximate induced matching oracle. For all induced matching oracles used in this work, one could think of $\alpha=\Omega(\poly(\eps))$ and $\beta=O(\poly(\eps))$.
\end{remark}

The black-box reduction needs induced matching oracles for multiple subgraphs. The initialization of an induced matching oracle can potentially be very costly -- in the worst case, it could take $\Omega(m)$ time if we picked some arbitrary subgraph. To avoid the potential cost, let's consider in what cases we get efficient initialization. Assume the induced matching oracle is already initialized on a graph $G=(V,E)$, and the subsequent initialization is on a subgraph $G[T]$ of $G$ induced by $T\subseteq V$. In this case, we can continue to work on $G$ because for every query $S\subseteq T$, $(G[T])[S]=G[S]$. In our framework, although the subgraphs involved are not immediately vertex-induced, their structures are good enough to be vertex-induced in the $\eps$-extension of the current graph. The $\eps$-extension is constructed by generating multiple copies of each vertex corresponding to its possible vertex dual values in our main algorithm, and by adding edges between the copies of original vertices such that the sum of the vertex dual values is approximately the edge weights. The formal definition is the following.

\begin{definition}[$\eps$-Extension] \label{def:extension}
     For an graph $G=(V,E)$ with aspect ratio $W$, denote $G^\eps(V_1\cup V_2, E^\prime)$ as its $\eps$-extension, where $V_1=\{v^1_i:\forall\;v\in V,i\in\{0,\eps\cdot W,2\eps\cdot W,\dots,W\}\}$, $V_2=\{v^2_i:\forall\;v\in V,i\in\{0,\eps\cdot W,2\eps\cdot W,\dots,W\}\}$ and $E^\prime=\{\{u^1_i,v^2_j\}:\{u,v\}\in E\land \bw(\{u,v\})\leq i+j\leq \bw(\{u,v\})+2\eps\cdot W\}$. 
     That is, for every edge 
     $\{u,v\}$ in the original graph, we add an edge $\{u^1_i,v^2_j\}$ to $E'$ for every $i,j\in\{0,\eps\cdot W,2\eps\cdot W\dots,W\}$ such that $\bw(\{u,v\})\leq i+j\leq \bw(\{u,v\})+2\eps\cdot W$. The number of vertices in $G^\eps$ is $O(n/\eps)$ and the number of edges is $O(m/\eps)$.
\end{definition}

Our key lemma states that with an oracle $\O$ initialized on the single graph $G^{\eps}$, we can compute an approximate MWM with a small number of calls to $\O$. We prove the lemma in~\cref{sec:weighted-key-lemma-proof}.

\begin{lemma}\label{lemma:structural:weighted}
    Given an $n$-node \textbf{general} (not necessarily bipartite) graph $G$ with aspect ratio $W$, and an $(\alpha,\beta)$-approximate induced matching oracle $\O$ initialized on $G^\eps$ with query time $\U(\alpha,\beta)$, there is an algorithm $\A$ that computes a $(1,\eps\cdot W\cdot n)$-approximate MWM on $G$ in time
    \[O(\alpha^{-1}\eps^{-123}\log(\alpha^{-1}\eps^{-1})\log^2n)\cdot \U(\alpha,O(\alpha\cdot\eps^{85})),\]
    with access to $G$ only through $\O$ and the adjacency matrix of $G$.
\end{lemma}

\section{Primal-Dual Framework for Non-Bipartite Approximate MWM}\label{sec:weighted-key-lemma-proof}

Now we give the algorithm $\A$ for~\cref{lemma:structural:weighted} that computes a $(1,\eps\cdot W\cdot n)$-approximate MWM on an $n$-node graph $G$ with an aspect ratio $W$.
The algorithm follows a primal-dual framework similar to~\cite{Pettie12,DuanP14} that maintains a matching $M$, a laminar family of blossoms $\Omega$, a set of dual variables $\by$ for vertices and $\bz$ for odd sets, an implicit set of edges $\Edel$, an explicit set of edges $\Mdel$\footnote{We point out that in our algorithm, at the time an edge is added to $\Mdel$, it is a matching edge. But in the end, the edge set $\Mdel$ is not necessarily a matching.}
and an explicit set of vertices $\Fdel$ with the properties below.
\begin{property}[$(\lambda,\theta,\eta)$-Relaxation]\label{property:weighted-structure}
Recall that we denote $s_{\{u,v\}}=y_u+y_v+\sum_{u,v\in B}z(B)$ as the dual of the edge $\{u,v\}$. We say that $\{M,\Omega,\by,\bz,\Edel,\Mdel,\Fdel\}$ satisfies $(\lambda,\theta,\eta)$-relaxation if $M$ respects $\Omega$, $\mu(\Edel),|\Mdel|,|\Fdel|\leq \theta\cdot n$, and,
\begin{enumerate}
    \item \textbf{Non-negativity}: $\by$-values are non-negative multiples of $\lambda\cdot W$ and $\bz$-values are non-negative multiples of $2\lambda\cdot W$. $\Omega$ includes all odd sets $B$ with $z_B>0$ and any non-trivial root blossom $B$ in $\Omega$ has $z_B>0$.
    \item \textbf{Approximate Domination}: For any edge $e\notin \Edel$, $s_e\geq \bw(e)$, i.e., $\by$ and $\bz$ are feasible dual variables for $G\setminus \Edel$.  
    \item \textbf{Approximate Tightness}: For any edge $e\in \left(M\cup \left(\cup_{B\in\Omega}E_B\right)\right)\setminus \Mdel$
    , $s_e\leq \bw(e)+2\lambda\cdot W$.
    For any edge $e\in\Mdel$, $s_e\leq \bw(e)+W$.
    \item \textbf{Free Vertices}: The $\by$-values of all free vertices except $\Fdel$ are exactly $\eta\cdot W$. For those free vertices in $\Fdel$, their $\by$-values are at most $W$.
\end{enumerate}
\end{property}

\begin{remark}
    Here we compare the set of variables we maintain to those of~\cite{Pettie12,DuanP14}. \cite{Pettie12} designs a primal-dual framework for computing \emph{exact} MWM based on exact MCM and asks for exact domination, exact tightness, and requires all free vertex values to be the same. If we require $\Edel,\Mdel,\Fdel=\emptyset$, $\by,\bz$ to be non-negative integers in \textbf{\textit{Non-negativity}}, and $s_e\leq \bw(e)$ instead of $s_e\leq \bw(e)+2\lambda\cdot W$ in \textbf{\textit{Approximate Tightness}}, then we get the property they maintain in their algorithm. \cite{DuanP14} designs another primal-dual framework for computing $(1-\eps,0)$-approximate MWM based on a different type of relaxation. For example, in \textbf{\textit{Approximate Domination}}, \cite{DuanP14} allows $s_e\geq \bw(e)-\delta$ for some $\delta$ depending the current scale. What we do is to  relax the constraint by allowing edges to have arbitrarily bad domination $s_e\geq \bw(e)-\infty$, but bounding the ``size'' of those bad edges ($\Edel$). Similarly, in \textbf{\textit{Approximate Tightness}}, their approximation for an edge depends on both the current scale and the last one that this edge was made a matching or blossom edge, but we instead use a fixed criterion and allow a small set of bad edges ($\Mdel$). To those readers who are familiar with~\cite{DuanP14}, on a high level, we are running just the first scale of their algorithm but allow more relaxation through $\Edel,\Mdel,\Fdel$. We need these additional relaxations to work with approximate induced matching oracles instead of the entire edge set, and the reason of stopping at the first scale is that we are just looking for an additive $(1,\eps\cdot W\cdot n)$-approximation rather than a multiplicative $(1-\eps,0)$-approximate one.
\end{remark}

\begin{lemma}\label{lemma:primal-dual-property:guarantee}
    Suppose the set $\{M,\Omega,\by,\bz,\Edel,\Mdel,\Fdel\}$ satisfies $(O(\eps),O(\eps),O(\eps))$-relaxation, then $M$ is a $(1,O(\eps)\cdot W\cdot n)$-approximate MWM.
\end{lemma}
\begin{proof}
By \textit{\textbf{Approximate Tightness}} and $|\Mdel|\leq O(\eps)\cdot n$,
\begin{equation}\label{eq:4.2:tightness}
\begin{aligned}
    \bw(M)
    &=\sum_{e\in M}\bw(e)=\sum_{e\in M\setminus \Mdel}\bw(e) + \sum_{e\in M\cap\Mdel} \bw(e)\geq \sum_{e\in M\setminus \Mdel}(s_e-O(\eps)\cdot W)+\sum_{e\in M\cap\Mdel} (s_e-W)\\
    &\geq \sum_{e\in M}s_e-O(\eps)\cdot W\cdot n=\sum_{u~\text{is matched}}y_u+\sum_{B\in\V_G}z_B\cdot |M\cap E_B|-O(\eps)\cdot W\cdot n.
\end{aligned}
\end{equation}
By \textit{\textbf{Free Vertices}} and $|\Fdel|\leq O(\eps)\cdot n$,
\begin{equation}\label{eq:4.2:free-vertices}
    \sum_{u~\text{is free}}y_u=\sum_{u\in \Fdel} y_u+\sum_{u\notin\Fdel~\text{is free}}y_u\leq W\cdot |\Fdel|+O(\eps)\cdot W\cdot n\leq O(\eps)\cdot W\cdot n.
\end{equation}
Since $M$ respects $\Omega$, $|M\cap E_B|=\left\lfloor\frac{|B|}{2}\right\rfloor$ for all $B\in \Omega$. By~\textit{\textbf{Non-negativity}}, $z_B=0$ for any odd set $B\in \V_G\setminus \Omega$, thus
\begin{equation}\label{eq:4.2:non-negativity}
\sum_{B\in\V_G}z_B\cdot |M\cap E_B|=\sum_{B\in \V_G}z_B\cdot \left\lfloor\frac{|B|}{2}\right\rfloor.
\end{equation}
Putting together~\cref{eq:4.2:tightness,eq:4.2:free-vertices,eq:4.2:non-negativity}, we have
\[\bw(M)\geq \sum_{u\in V}y_u+\sum_{B\in\V_G}z_B\cdot \left\lfloor\frac{|B|}{2}\right\rfloor-O(\eps)\cdot W\cdot n.\]
Now consider any MWM $M^\prime$. By~\textit{\textbf{Approximate Domination}} and $\mu(\Edel)\leq O(\eps)\cdot n$, 
\begin{equation}
\begin{aligned}
    \bw(M^\prime)&=\sum_{e\in M^\prime\setminus \Edel} \bw(e)+\sum_{e\in M^\prime\cap \Edel}\bw(e)\leq \sum_{e\in M^\prime\setminus \Edel} s_e+W\cdot \mu(\Edel)\leq \sum_{e\in M^\prime\setminus \Edel} s_e+O(\eps)\cdot W\cdot n\\
    &\leq \sum_{u~\text{matched in}~M^\prime\setminus\Edel} y_u+\sum_{B\in\V_G}z_B\cdot |(M^\prime\setminus \Edel)\cap E_B|+O(\eps)\cdot W\cdot n.
\end{aligned}
\end{equation}
By \textbf{\textit{Non-negativity}}, $y_u\geq 0$ and $z_B\geq 0$. Since $M^\prime\setminus \Edel$ is a matching, $|(M^\prime\setminus \Edel)\cap E_B|\leq \left\lfloor\frac{|B|}{2}\right\rfloor$.
Therefore,
\[\bw(M^\prime)\leq \sum_{u\in V}y_u+\sum_{B\in\V_G}z_B\cdot \left\lfloor\frac{|B|}{2}\right\rfloor+O(\eps)\cdot W\cdot n\leq \bw(M)+O(\eps)\cdot W\cdot n.\]
\end{proof}

In the remainder of the section, we will give an algorithm to find a set $\{M,\Omega,\by,\bz,\Edel,\Mdel,\Fdel\}$ satisfying $(O(\eps),O(\eps),O(\eps))$-relaxation based on an $(\alpha,\beta)$-approximate induced matching oracle.

\subsection{Algorithm Description: Framework}\label{sec:weighted-algorithm-description-outer-loop}

Starting from here, we show our algorithm,~\cref{alg:outer-loop}, for finding a set $\{M,\Omega,\by,\bz,\Edel,\Mdel,\Fdel\}$ satisfying $(O(\eps),O(\eps),O(\eps))$-relaxation. 
The algorithm has $1/(2\eps)$ rounds of computation\footnote{Assume w.l.o.g.\ that $1/(2\eps)$ is an integer.}. In each round, we consider the contracted tight subgraph $\Gtight/\Omega$ w.r.t.\ the current $\{M,\Omega,\by,\bz\}$.

\begin{definition}[Tight Subgraph]
    The tight subgraph of $G$ w.r.t.\ $\{M,\Omega,\by,\bz\}$ is denoted by $\Gtight(V,\Etight)$ where $\Etight=\{e\in E \setminus \left(\cup_{B\in\Omega} E_B\right):\bw(e)\leq s_e\leq \bw(e)+2\eps\cdot W\} \cup \left(M\cup\left(\cup_{B\in\Omega} E_B\right)\right)$.
\end{definition}

We view the contracted tight subgraph $\Gtight/\Omega$ as an \textit{unweighted} graph. In a single round, the algorithm computes information about $\Gtight/\Omega$ to adjust our set $\{M,\Omega,\by,\bz,\Edel,\Mdel,\Fdel\}$ through three steps: \texttt{AUGMENTATION-AND-BLOSSOM-FORMATION}, \texttt{DUAL-ADJUSTMENT}, and\\ \texttt{BLOSSOM-DISSOLUTION}. It then leads to a new $\Gtight/\Omega$, which is processed in the next round.

\begin{algorithm}[!ht]
    \SetKwInput{KwData}{Input}
    \caption{Our Framework} \label{alg:outer-loop}
    \SetEndCharOfAlgoLine{}
    
  \SetEndCharOfAlgoLine{}
  \SetKwInput{KwData}{Input}
  \SetKwComment{Comment}{/* }{ */}
  \SetKwProg{KwProc}{function}{}{}
  \SetKwFunction{Initialize}{Initialize}
  \SetKwFunction{AugmentationAndBlossomFormation}{AUGMENTATION-AND-BLOSSOM-FORMATION}
  \SetKwFunction{DualAdjustment}{DUAL-ADJUSTMENT}
  \SetKwFunction{BlossomDissolution}{BLOSSOM-DISSOLUTION}

  \Comment{Initialization of the set $\{M,\Omega,\by,\bz,\Edel,\Mdel,\Fdel\}.$}
    $M,\Edel,\Mdel,\Fdel\gets\emptyset.$\;
    $\Omega\gets\{\{u\}:u\in V\}.$\;
    $\forall u\in V,~y_u\gets W/2.$\;
    $\forall B\in\V_G, z_B\gets 0.$\;
    \Comment{$\frac1{2\eps}$ rounds of computation.}
    \For{round $t=1,2,\dots,1/(2\eps)$}{
        Consider $\Gtight$ w.r.t.\ $\{M,\Omega,\by,\bz\}$.\;
        Denote $\O$ as an $(\alpha,\beta)$-approximate induced matching oracle initialized on $\B_{\Gtight/\Omega}$.\;
        Denote $\M$ as an adjacency matrix oracle initialized on $\B_{\Gtight/\Omega}$.\;
        \AugmentationAndBlossomFormation{}.\;
        \DualAdjustment{}.\;
        \BlossomDissolution{}.\;
    }
\end{algorithm}

To efficiently compute information about the sequence of tight subgraphs, the algorithm will use $(\alpha,\beta)$-approximate induced matching oracles and adjacency matrix oracles initialized on the bipartization $\B_{\Gtight/\Omega}$. The adjacency matrix oracle is defined as follows.
\begin{definition}[Adjacency Matrix Oracle]
An adjacency matrix oracle initialized on an $n$-node graph $G=(V,E)$, when given a query $\{u,v\}\subseteq V\times V$, returns \texttt{True} if $\{u,v\}\in E$ and \texttt{False} otherwise.
\end{definition}
Note that in~\cref{lemma:structural:weighted} we only needed the assumption that our oracle is initialized on a single graph $G^\eps$. Here, we first show an algorithm with the stronger assumption that we have an initialized oracle for every contracted tight subgraph. We will remove this assumption in~\cref{sec:induced-matching-oracle-simulation} by showing how a single oracle initialized on $G^\eps$ can efficiently answer queries on any $\B_{\Gtight/\Omega}$. In~\cref{sec:implement-adjacenct-matrix} we show how to implement the adjacency matrix oracle.

\subsection{High Level Description of One Round of Computation and Key Subroutine}
The initialization in~\cref{alg:outer-loop} guarantees that the initial set 
$\{M,\Omega,\by,\bz,\Edel,\Mdel,\Fdel\}$ satisfies $(\eps,0,1/2)$-relaxation
by definition. 
In each round, we compute information about the current graph $\Gtight$ (viewed as an unweighted graph) and use this to adjust $\{M,\Omega,\by,\bz,\Edel,\Mdel,\Fdel\}$. As we will show later in~\cref{sec:single-round-improvement}, the set satisfies $(\eps,O(\eps^2\cdot t),1/2-\eps\cdot t))$-relaxation at the end of the $t$-th round, 
thus satisfying $(\eps,O(\eps),0)$-relaxation eventually after $1/(2\eps)$ rounds. Note that for the three parameters in the $(\lambda,\theta,\eta)$-relaxation, $\lambda$ reamins the same, $\theta$ keeps relaxing per round; thus, the focus is to make $\eta$ tighter and tighter. As a reminder, $\eta$ manages the $\by$-values of the free vertices not in $\Fdel$. This is consistent with the framework in~\cite{Pettie12,DuanP14}, and we will introduce our computation per round based on their algorithm.

In~\cite{Pettie12}, the computation per round is fairly simple, since they use an exact MCM algorithm as a subroutine. They find an MCM on the contracted tight subgraph, expand the matching back to $G$, contract the blossoms, and adjust the dual. To explain the idea behind this, let's first imagine that we are working on a bipartite graph. To achieve the goal of decreasing the $\by$-values of the free vertices while maintaining the domination and tightness, we can label a vertex as an ``inner vertex'' (resp., ``outer vertex'') if there is an odd-length (resp., even-length) alternating path linking it to a free vertex in the tight subgraph. The parity is consistent \textit{when there are no augmenting paths and the graph is bipartite}. We then could decrease the vertex dual for all ``outer vertices'' and increase the dual for all ``inner vertices'' -- this causes the dual values of the free vertices drop, all edges within the tight graph to remain tight and some non-tight edges on the boundary of the tight subgraph may become tight; thus the tight subgraph expands, and we continue to the next round.

In non-bipartite graphs, odd cycles make it unclear how to only modify the vertex dual when the parity could be inconsistent. The way to fix it is to contract the blossom and work with the contracted graph (which will be ``close'' to bipartite in the sense that eventually on the contracted graph, all the alternating paths between a vertex to any free vertex have the same length parity). We similarly define the ``inner vertices'' and ``outer vertices'' and adjust the vertex or odd-set dual accordingly. An additional step will be done in the non-bipartite graphs, which is to dissolve the non-trivial root blossom with zero $\bz$-values so we can keep the non-negativity of the duals. The remaining ideas are the same. The arguably most important point here is the property of no augmenting paths on the contracted tight subgraph, so that the parity is consistent for us to do dual adjustment.

In~\cite{DuanP14}, they relax the exact dominance and tightness, thus having more flexibility. Instead of finding an exact MCM on 
$\Gtight/\Omega$, which might take too much time, they find a maximal set of augmenting paths. Through a different definition of tight subgraphs, they make sure that every edge in the augmenting paths is removed from the tight subgraphs after the augmentation, guaranteeing the no augmenting paths property.

In our case, either an exact MCM or an exact maximal set of augmenting paths is hard to find with the oracle type we have. We instead use a slightly weaker black box that allows us to modify the graph by deleting a small set of edges ($\Edel$, $\Mdel$) and vertices ($\Fdel$), so that 
the remaining graph satisfies the desired properties. Our black box is based on recent progress on parallel DFS~\cite{FischerMU22,MitrovicMSS25}. A key property of their algorithm is that, after deleting a relatively small number of edges and vertices, there are no \textit{short} augmenting paths on the graph. We make some other deletions based on additional structure maintained by their algorithm to ensure that there are no long augmenting paths either. The following lemma formally states the sets that we can compute in an unweighted graph; this will be the main building block of a single round of our primal-dual framework.

\begin{lemma}\label{lemma:structural}
    For any $n$-node graph $G=(V,E)$, given an $(\alpha,\beta)$-approximate induced matching oracle initialized on $\B_G$ with query time $\U_\O$, and an adjacency matrix oracle of $G$
    with query time $\U_\M$,
    there is an algorithm $\A$ in time
    \[O(\beta^{-19/13}\alpha^{6/13}\log n\log (1/\beta))\cdot \U_\O+O(\alpha/\beta)\cdot n\cdot \U_\M,\] that given an initial matching $M_0$ and access to $G$ only through the oracles, partitions $V$ into three subsets $\Vin,\Vout,\Vunfound$, finds a laminar family of blossoms $\Omega$ and a matching $M$ that extends $M_0$ and respects $\Omega$,
    a set of free vertices $\Fdel$ w.r.t.\ M with size $|\Fdel|=O((\beta/\alpha)^{2/13})\cdot n$, and implicitly defines an edge set $\Edel$ s.t.\ $\mu(\Edel)=O((\beta/\alpha)^{1/13})\cdot n$ that satisfy
    \begin{enumerate}
        \item For all nontrivial $B\in\Omega$, $B\subseteq \Vout$.
        \item For any edge $\{u,v\}\in M$ s.t.\ $\Omega(u)\neq \Omega(v)$, either both of its end points are in $\Vunfound$, or one of them is in $\Vin$ and the other is in $\Vout$.\label{item:2}
        \item No edge $\{u,v\}$ in $E\setminus \Edel$ s.t.\ $\Omega(u)\neq \Omega(v)$ has both of its end points in $\Vout$, or one in $\Vout$ and one in $\Vunfound$.\label{item:3}
        \item All free vertices w.r.t.\ $M$ except $\Fdel$ are in $\Vout$, and $\Fdel$ are in $\Vunfound$.
    \end{enumerate}
\end{lemma}

For intuition about the properties above, imagine we apply the lemma on the contracted tight subgraph. One could think of $\Vin$ (resp., $\Vout$) as the vertices reached by odd-length (resp., even-length) alternating paths from free vertices after further contracting $\Omega$
and ``deleting'' the edges in $\Edel$ and the free vertices in $\Fdel$. The set $\Vunfound$ corresponds to vertices that become ``unreachable'' from free vertices via alternating paths in the contracted tight subgraph after contraction and deletion.

\paragraph{Comparison to prior work}
The proof of~\cref{lemma:structural} is quite similar to the parallel DFS/matching algorithms from~\cite{FischerMU22,MitrovicMSS25}, with some modifications to ensure additional properties and to recast their algorithm in terms of an induced matching oracle. We defer the proof to~\cref{sec:structural-lemma-proof}.

\cite{HuangS23} has a similar primal-dual framework which also builds upon parallel DFS/matching algorithms, and computes a $(1-\eps,0)$-approximate MWM in semi-streaming, parallel and distributed models. However, their algorithm needs the assumption of $O(\log^3 n/\eps)$ weak diameters and is not implemented in the dynamic setting.

We also point out that a very recent paper~\cite{MitrovicS25} similarly relies on parallel DFS/matching algorithms and approximate induced matching oracles, but to find $(1-\eps,0)$-approximate MCM in non-bipartite graphs. Because their paper is only about MCM (i.e., unweighted matching), the no-short-augmenting-paths property from the parallel DFS suffices to find a large matching. In contrast, inside our primal-dual framework for weighted matching, each round needs to remove \emph{every} augmenting path. More generally, even if we have a black-box algorithm for finding $(1-\eps,0)$-approximate MCM on non-bipartite graphs, this on its own does not seem strong enough to plug into a primal-dual framework. The properties we need from \Cref{lemma:structural} are thus somewhat stronger. 

\paragraph{Handling Large Blossoms} We will show how to use \Cref{lemma:structural} as a black box to implement one round of our primal-dual framework. Before diving into the details, we want to point out another observation we made to simplify the algorithm. Since we are working with contracted graphs, large blossoms are difficult to deal with because the number of inter-blossom edges is large. To keep the blossom-size small, we observe that, for a newly-discovered non-trivial root blossom, we can dissolve it and adjust the vertex dual inside so that $\bs$-values are non-decreasing. Moreover, depending on whether this blossom is free or not, we can move at most one free vertex into $\Fdel$ or at most one matched edge into $\Mdel$ for the relaxation to hold. Therefore, by dissolving such blossoms with large sizes, every blossom we maintain at the end of each round will have $\poly(\eps^{-1})$ size.

\subsection{A Single Round of the Primal-Dual Algorithm}

For a single round of the primal-dual framework, we assume that we are given the current set $\{M,\Omega,\by,\bz,\Edel,\Mdel,\Fdel\}$, an $(\alpha,\beta)$-approximate induced matching oracle and an adjacency matrix oracle, both initialized on $\B_{\Gtight/\Omega}$. We then proceed in the following steps.
\paragraph{Augmentation and Blossom Formation} Apply~\cref{lemma:structural} on $\Gtight/\Omega$ with the initial matching as $M/\Omega$ and denote the set $\{\Vin,\Vout,\Vunfound,\Omega^\prime,\Edel^\prime,\Fdel^\prime\}$ as the objects listed in~\cref{lemma:structural} on $\Gtight/\Omega$.
We first expand $\Vin,\Vout,\Vunfound,\Omega^\prime,\Edel^\prime$ according to $\Omega$ as follows.
\begin{align*}
    \Vin&\gets \cup_{B\in \Vin} B;\\
    \Vout&\gets \cup_{B\in \Vout} B;\\
    \Vunfound&\gets \cup_{B\in \Vunfound} B;\\
    \Omega^\prime&\gets \{\cup_{B\in B^\prime} B:B^\prime\in\Omega^\prime\};\\
    \Edel^\prime&\gets \{\{u,v\}\in E:\{\Omega(u),\Omega(v)\}\in \Edel^\prime\}.
\end{align*}
In other words, the union of the corresponding blossoms replaces a set of contracted vertices, and the union of all edges between the two blossoms replaces an edge on the contracted graph. For the contracted vertices in $\Fdel^\prime$, we replace them with the corresponding free vertex within the blossom,
\[\Fdel^\prime\gets\{\text{The free vertex}~v\in B:B\in\Fdel^\prime\}.\]
For $M^\prime$, consider the set of augmenting paths $P^\prime=M^\prime\oplus (M/\Omega)$, extend it to $P$ in $G$ according to~\cref{lemma:expanding-augmenting-paths} and set
\[M^\prime\gets P\oplus M.\]
Thus $M^\prime$ respects $\Omega^\prime\cup\Omega$ and extends $M$. Then update the set $\{M,\Omega,\Edel,\Fdel\}$ as follows.
\begin{align*}
    M&\gets M^\prime.\\
    \Omega&\gets \Omega^\prime\cup\Omega.\\
    \Edel&\gets \Edel\cup \Edel^\prime.\\
    \Fdel&\gets \Fdel\cup\Fdel^\prime.
\end{align*}

\paragraph{Dual Adjustment}
In this step, we update the dual variables $\{\by,\bz\}$.
\begin{align*}
y_u&\gets y_u-\eps\cdot W,\quad \forall\;u\in \Vout.\\
y_u&\gets y_u+\eps\cdot W,\quad \forall\;u\in \Vin.\\
z_B&\gets z_B+2\eps\cdot W,\quad \forall\;\text{non-trivial root blossom}~B~\text{in}~\Omega~\text{and}~B\subseteq\Vout.\\
z_B&\gets z_B-2\eps\cdot W,\quad \forall\;\text{non-trivial root blossom}~B~\text{in}~\Omega~\text{and}~B\subseteq\Vin.
\end{align*}
\paragraph{Blossom Dissolution}
Repeat the following two processes until no blossom satisfies either constraint. Firstly, remove any non-trivial root blossom $B$ in $\Omega$ with $z_B=0$; this preserves \Cref{property:weighted-structure}. Secondly, for any root blossom $B$ of size $|B|\geq \eps^{-2}$, set $y_u\gets y_u+z_B/2,~\forall u\in B$ and $z_B\gets 0$.  Observe that this does not affect $s_e$ for $e \in G[B]$. We then make the following changes to perserve \Cref{property:weighted-structure} (these are analyzed later in the section). Case 1: if the dissolved blossom $B$ is a free blossom, 
then the dissolution will increase the $\by$-value of the only free vertex $r$ of $B$, so we update $\Fdel \gets \Fdel\cup\{r\}$. Case 2: if the dissolved blossom $B$ is matched, then the dissolution will increase the $\bs$-value of the unique matched edge $e$ connecting $B$ to another blossom, so if $s_e \geq$ $\bw(e)+2\eps\cdot W$, we update $\Mdel\gets\Mdel\cup \{e\}$.

\begin{observation}\label{lemma:framework:blossom-size}
    At the end of each round, the largest blossom in $\Omega$ has size at most $O(\eps^{-2})$.
\end{observation}

\subsection{Improvement in a Single Round}\label{sec:single-round-improvement}

We now show the progress made in each round towards finding a $(O(\eps),O(\eps),O(\eps))$-relaxation.
\begin{lemma}\label{lemma:maintain-key-property}
    For $\beta=O(\alpha\cdot \eps^{78})$, at the end of the $t$-th round, where $t=0,1,\dots,1/(2\eps)$, the set $\{M,\Omega,\by,\bz,\Edel,\Mdel,\Fdel\}$ satisfies $O(\eps,O(\eps^2\cdot t),1/2-\eps\cdot t)$-relaxation.
\end{lemma}

\begin{proof} We prove the properties of \Cref{property:weighted-structure} one by one.

\paragraph{Non-negativity} During initialization, $\by$-values are set to be $\frac W2$, and they decrease by at most $\eps\cdot W$ per round. Since $t\leq 1/(2\eps)$, $\by$-values are non-negative multiples of $\eps\cdot W$. During initialization, $\bz$-values are set to be 0 and will be multiples of $2\eps\cdot W$ throughout the algorithm. Since in the blossom dissolution step, the algorithm dissolves all non-trivial root blossoms with zero $\bz$-value, every non-trivial root blossom has a positive $\bz$-value at the end of each round. Also, since all newly added blossoms are in $\Vout$, the dual adjustment step 1) only increases the $\bz$-value of the blossom in $\Omega$, so $\Omega$ includes all odd sets with positive $\bz$-values; and 2) only decreases the $\bz$-value of a blossom whose $\bz$-value is positive, which is at least $2\eps\cdot W$, thus $\bz$-values are non-negative multiples of $2\eps\cdot W$.

\paragraph{Approximate Domination} During initialization, $\by$-values are set to be $\frac W2$ and $\bz$-values are set to be 0, thus for any edge $e\in E$, its $\bs$-value is $W\geq \bw(e)$. First, consider how dual adjustment changes the $\bs$-value in each round. For an edge $e = \{u,v\}$ in $\Vout$ such that $\Omega(u)=\Omega(v)$, both $y_u$ and $y_v$ decrease by $\eps\cdot W$ and the $\bz$-value of the root blossom that covers edge $e$ increases by $2\eps\cdot W$, so $s_e$ remains unchanged. For the same reason, the $\bs$-value of any edge $e=\{u,v\}$ in $\Vin$ such that $\Omega(u)=\Omega(v)$ does not change. For any inter-blossom edge $e$ such that $s_e\geq \bw(e)$ before this round, if $e\notin \Etight$, i.e., $s_e>\bw(e)+2\eps\cdot W$, then in the dual adjustment step, the $\by$-values of both endpoints decrease at most $\eps\cdot W$, thus its $\bs$-value decreases at most $2\eps\cdot W$, and still at least $\bw(e)$ after this step. For any inter-blossom edge $e\in \Etight$, according to \cref{item:3} in~\cref{lemma:structural}, and how we expand $\Edel^\prime$ and set the new $\Omega$, no edge in $(\Etight\setminus \Edel^\prime)/\Omega$ have both of its end points in $\Vout$, or one in $\Vout$ and one in $\Vunfound$ thus their $\bs$-values do not decrease. To summarize, any edge $e$ with $\bs$-values at least $\bw(e)$ before this step and less than $\bw(e)$ after is included in $\Edel^\prime$ this round, and thus is included in $\Edel$. Now, consider how the blossom dissolution step would change the $\bs$-values. The $\bs$-values of the edges within that dissolved blossom do not change for a similar reason as before. For inter-blossom edges related to the dissolved blossom, its $\bs$-value only increases because the $\bz$-values are non-negative. Thus, this step does not create edges with $\bs$-values at least $\bw(e)$ before this step and less than $\bw(e)$ after.

\paragraph{Approximate Tightness}
During the initialization, there are no matching edges or blossom edges, and $\Mdel=\emptyset$, so approximate tightness holds. In the augmentation and blossom formation, any new matching edge or blossom edge $e$ must be in $\Etight$, and the only edges in $\Etight$ with $\bs$-values larger than $2\eps\cdot W$ are those in $\Mdel$. Now consider how the dual adjustment changes the $\bs$-values. For blossom edges, for the same reason as in the proof for approximate domination, the $\bs$-values do not change. For inter-blossom edges in $M$, according to \cref{item:2} in \cref{lemma:structural} and how we set the new $M$ and $\Omega$, for edges in $M/\Omega$, either both of its end points are in $\Vunfound$, or one of them is in $\Vin$ and the other is in $\Vout$, thus their $\bs$-values do not change. In the blossom dissolution step, similarly, the $\bs$-values of the matched edges within that dissolved blossom do not change. The only matched edge affected is the one between this dissolved blossom and some other blossom, which is collected by $\Mdel$. To complete the proof of approximate tightness, we argue that for any  $e\in\Mdel$, $s_e\leq \bw(e)+W$. Since we run the blossom dissolution step every round, the blossom dissolved due to its size must be a newly formed one, and its $\bz$-value is $2\eps\cdot W$. Thus, the change in the $\bs$-values of the affected matched edges, those added to $\Mdel$, is $\eps\cdot W$. Within $t\leq 1/(2\eps)$ rounds, its $\bs$-value is at most $\bw(e)+2\eps\cdot W+t\cdot \eps\cdot W\leq \bw(e)+W$ for $\eps\leq \frac{1}{4}$.

\paragraph{Free Vertices} After initialization, all vertices have the same $\by$-value of $\frac W2$. In any one round, the $\by$-value of any vertex can decrease by at most $\eps\cdot W$, and since free vertices except $\Fdel$ are in $\Vout$ by definition, the $\by$-value of those free vertices decreases by precisely this maximum amount. The only way the $\by$-value of a free vertex $v$ can \emph{increase} is during the blossom dissolution phase, but in this case $v$ is added to $\Fdel$. Finally, observe that in augmentation and blossom formation the current matching extends the previous matching according to how we set the new $M$, so throughout the algorithm any free vertex has been free from the beginning. In conclusion, any vertex that is free at the beginning of the $t$-th round and not in $\Fdel$ has decreased its $\by$-value by $\eps\cdot t\cdot W$ as desired. And for free vertices in $\Fdel$, its $\by$-value can increase by at most $\eps\cdot W$ in the blossom dissolution of each step, thus at most $W$ at the beginning of the $t$-th round.

\paragraph{Matching versus Blossoms} The algorithm always has $M$ respecting $\Omega$.

\paragraph{Bounding $\mu(\Edel),|\Mdel|$ and $|\Fdel|$}
By~\cref{lemma:structural}, in \AugmentationAndBlossomFormation,
$\mu(\Edel^\prime)=O((\beta/\alpha)^{1/13})\cdot n$. According to how we extend $\Edel^\prime$ and~\cref{lemma:framework:blossom-size}, after extension, $\mu(\Edel^\prime)=O(\eps^{-4}(\beta/\alpha)^{1/13})\cdot n$. Thus, at the beginning of the $t$-th round, $\mu(\Edel)=O(t\cdot \eps^{-4}(\beta/\alpha)^{1/13})\cdot n=O(t\cdot \eps^2)\cdot n$.

In \BlossomDissolution, $|\Mdel|$ increases by at most one only when a blossom with non-zero $\bz$-value and with size at least $\eps^{-2}$ are dissolved. These blossoms have to be newly formed root blossoms thus are disjoint, and $|\Mdel|$ increases by $O(\eps^{2})\cdot n$ per round. Thus, at the beginning of the $t$-th round, $|\Mdel|=O(t\cdot \eps^2)\cdot n$

By~\cref{lemma:structural}, in \AugmentationAndBlossomFormation, $|\Fdel^\prime|= O((\beta/\alpha)^{2/13})\cdot n$ and in \BlossomDissolution, it collects at most $O(\eps^2)\cdot n$ vertices since only blossoms with size at least $O(\eps^{-2})$ are dissolved, and at most one vertex is added to $\Fdel$. Thus, at the beginning of the $t$-th round, $|\Fdel|= O((\eps^2+(\beta/\alpha)^{2/13})\cdot t)\cdot n=O(t\cdot \eps^2)\cdot n$.

\end{proof}

Based on ~\cref{alg:outer-loop} and the analysis above, we can now prove a weaker version of ~\cref{lemma:structural:weighted} that does not account for the time to initialize the oracles; the initialization is then handled in the next section.

\begin{lemma}\label{lemma:structural:weighted:weaker-version}
    Given an $(\alpha,\beta)$-approximate induced matching oracle $\O$ with $\U_\O(n,m,\alpha,\beta)$ query time, and an adjacency matrix oracle $\M$ with $\U_\M(n,m)$ query time on an $n$-node $m$-edge \textbf{bipartite} graph, \cref{alg:outer-loop} computes a $(1,\eps\cdot W\cdot n)$-approximate MWM on an $n$-node $m$-edge \textbf{general} graph $G$ (not necessarily bipartite) with aspect ratio $W$ in time
    \[O(\alpha^{-1}\eps^{-115}\log(\alpha^{-1}\eps^{-1})\log n)\cdot \U_\O(O(n),O(m),\alpha,O(\alpha\cdot\eps^{78}))+O(\eps^{-79})\cdot n\cdot \U_\M(n,m).\]
    assuming at the beginning of the $t$-th round, $\O$ and $\M$ are already initialized (with no cost to the algorithm) on the corresponding graph $\B_{\Gtight/\Omega}$.
\end{lemma}

\begin{proof}
    By~\cref{lemma:maintain-key-property}, the set $\{M,\Omega,\by,\bz,\Edel,\Mdel,\Fdel\}$ maintained by~\cref{alg:outer-loop} satisfies $(\eps,O(\eps),0)$-relaxation when the algorithm terminates. Thus, by~\cref{lemma:primal-dual-property:guarantee}, $M$ is a $(1,O(\eps)\cdot W\cdot n)$-approximate MWM.

    The algorithm has $1/(2\eps)$ rounds. For each round, it sets $\beta=O(\alpha\cdot \eps^{78})$ according to~\cref{lemma:maintain-key-property} and  invokes~\cref{lemma:structural} with run time
    \[O(\alpha^{-1}\eps^{-114}\log(\alpha^{-1}\eps^{-1})\log n)\cdot \U_\O(O(n),O(m),\alpha,O(\alpha\cdot\eps^{78}))+O(\eps^{-78})\cdot n\cdot \U_\M(n,m).\]
    The time taken for transforming the output is $O(\eps^{-2})\cdot n$, for \DualAdjustment and \BlossomDissolution is $O(n)$, thus dominated by the above term. The total runtime for $1/(2\eps)$ rounds is thus
    \[O(\alpha^{-1}\eps^{-115}\log(\alpha^{-1}\eps^{-1})\log n)\cdot \U_\O(O(n),O(m),\alpha,O(\alpha\cdot\eps^{78}))+O(\eps^{-79})\cdot n\cdot \U_\M(n,m).\]
\end{proof}

\subsection{Simulating Induced Matching Oracle on Contracted Tight Subgraphs}\label{sec:induced-matching-oracle-simulation}

\Cref{alg:outer-loop} uses approximate induced matching oracles initialized on the bipartization of the contracted tight subgraph $\B_{\Gtight/\Omega}$. However, the contracted tight subgraph changes every round and is not a vertex-induced subgraph of the main graph $G$. To avoid the cost of repeatedly initializing each of these oracles, we will show that the induced matching queries on \emph{every} $\B_{\Gtight/\Omega}$ can be answered efficiently given an induced matching oracle initialized on the single graph $G^\eps$ (see \Cref{def:extension}). 

\begin{observation}\label{lemma:property-of-inter-blossom-edges}
    At the beginning of the $t$-th round, the set $\{\Omega,\by,\bz\}$ satisfies that for any edge $\{u,v\}\in G$ such that $\Omega(u)\neq \Omega(v)$, $s_{\{u,v\}}=y_u+y_v$.
\end{observation}

\begin{proof}
    By definition, $s_{\{u,v\}}=y_u+y_v+\sum_{u,v\in B} z_B$. Consider any edge $\{u,v\}$ in $G$ s.t.\ $\Omega(u)\neq \Omega(v)$. By \textit{\textbf{Non-negativity}}, any odd set $B$ with $z_B>0$ is included in $\Omega$. Thus, for any blossom $B$ s.t.\ $u,v\in B$, $z_B=0$, and $s_{\{u,v\}}=y_u+y_v$.
\end{proof}

Recall that the edge set of $\Gtight$ contains two parts, $\{e\in E \setminus \left(\cup_{B\in\Omega} E_B\right):\bw(e)\leq s_e\leq \bw(e)+2\eps\cdot W\}$ and $M\cup\left(\cup_{B\in\Omega} E_B\right)$. Any edge in $\Gtight/\Omega$ corresponds to an edge $\{u,v\}\in G$ such that $\Omega(u)\neq \Omega(v)$. Thus, by~\cref{lemma:property-of-inter-blossom-edges}, the first part can be relaxed to $\{\{u,v\}\in E:\bw(\{u,v\})\leq y_u+y_v\leq \bw(\{u,v\})+2\eps\cdot W\}$, and the bipartization of this edge set is a vertex-induced subgraph of $G^\eps$. From now on, we will work with $\Gtight^\prime=(V,\Etight^\prime)$ instead of $\Gtight$, where $\Etight^\prime=\{\{u,v\}\in E:\bw(\{u,v\})\leq y_u+y_v\leq \bw(\{u,v\})+2\eps\cdot W\}\cup M$.

\begin{observation}\label{lemma:tight-equivalence}
    At the beginning of the $t$-th round, $\Etight/\Omega=\Etight^\prime/\Omega$ w.r.t.\ the set $\{M,\Omega,\by,\bz\}$.
\end{observation}

\begin{proof}
    The vertex sets are the same. For any edge $\{B_1,B_2\}$ in $\Etight/\Omega$, there exists an edge $\{u,v\}\in G$ such that $\Omega(u)\neq \Omega(v)$ and either $\bw(\{u,v\})\leq s_{\{u,v\}}\leq \bw(\{u,v\})+2\eps\cdot W$, or $\{u,v\}\in M\cup\left(\cup_{B\in\Omega} E_B\right)$. In the first case, by~\cref{lemma:property-of-inter-blossom-edges}, we have $\bw(\{u,v\})\leq y_u+y_v\leq \bw(\{u,v\})+2\eps\cdot W$ thus $\{u,v\}\in\Etight^\prime$. In the second case, since $\Omega(u)\neq\Omega(v)$, $\{u,v\}\in M\subseteq \Etight^\prime$. Thus $\{B_1,B_2\}\in\Etight^\prime/\Omega$ and $\Etight/\Omega\subseteq \Etight^\prime/\Omega$. Through a similar argument, we can prove that $\Etight^\prime/\Omega\subseteq \Etight/\Omega$, thus $\Etight/\Omega=\Etight^\prime/\Omega$.
\end{proof}

\cref{lemma:tight-equivalence} shows that an induced matching query on $\B_{\Gtight/\Omega}$ is equivalent to an induced matching query on $\B_{\Gtight^\prime/\Omega}$, which can be answered efficiently given an induced matching oracle initialized on $\B_{\Gtight^\prime}$ using~\cref{alg:reduction:contracted-to-uncontracted} and its analysis~\cref{lemma:reduction:contracted-to-uncontracted}.

\begin{algorithm}[!ht]
    \SetKwInput{KwData}{Input}
    \caption{Simulating Induced Matching Queries on the Contracted Graph} \label{alg:reduction:contracted-to-uncontracted}
    \SetEndCharOfAlgoLine{}
    
  \SetEndCharOfAlgoLine{}
  \SetKwComment{Comment}{/* }{ */}
  
    \KwData{An $(\alpha,\beta)$-approximate induced matching oracle $\O$ on $\B_G$, a laminar family of odd sets $\Omega$ s.t.\ the largest odd set in $\Omega$ has size at most $\gamma$, and an induced matching query $T^1\cup T^2$ on $\B_{G/\Omega}$ where $T^1$ (resp., $T^2$) is in the left (resp., right) part.}
  
    \For{$k=1,2,\dots,O(\gamma^2\log n)$}{
        $S^1,S^2\gets \emptyset$.\;
        \lFor{$\forall~B\in T^1$}{
            Independently sample a vertex $u\in B$ u.a.r.\ and let $S^1\gets S^1\cup \{u\}$.
        }
        \lFor{$\forall~B\in T^2$}{
            Independently sample a vertex $u\in B$ u.a.r.\ and let $S^2\gets S^2\cup \{u\}$.
        }
        $M^k\gets \O(S^1\cup S^2)$.\;
    }
    $M^\prime\gets$ the largest matching among $\{M^k\}$.\;
    \textbf{return} $M^\prime/\Omega$.
\end{algorithm}

\begin{lemma}\label{lemma:reduction:contracted-to-uncontracted}
    Let $G$ be any graph, and say that we are given an $(\alpha,\beta)$-approximate induced matching oracle $\O$ that is initialized on the bipartization $\B_G$ and has query time $\U$. Say that we are also given a laminar family of odd sets $\Omega$ of $G$ s.t.\ the largest odd set in $\Omega$ has size at most $\gamma$. Then, an $(\Omega(\alpha/\gamma^2),\beta\gamma)$-approximate induced matching query on $\B_{G/\Omega}$ can be answered (w.h.p.\ against an adaptive adversary) in time
    \[O(\gamma^2\log n)\cdot \U\]
    using~\cref{alg:reduction:contracted-to-uncontracted}.

\end{lemma}

\begin{proof}
   The runtime of~\cref{alg:reduction:contracted-to-uncontracted} is clearly $O(\gamma^2\log n)\cdot \U$. We now prove that it returns an $(\Omega(\alpha/\gamma^2),\beta\gamma^2)$-approximate induced matching w.h.p.\ against an adaptive adversary.

    Consider any MCM on $\B_{G/\Omega}[T^1\cup T^2]$, there is a corresponding matching $M\in\B_G$ with the same size ($M$ is not necessarily an MCM in $\B_G$). Let $X_e$ be the indicator function of whether edge $e$ is kept in $\B_G[S^1\cup S^2]$ after the sampling. The sampling process ensures that for any edge $e\in M$, w.p.\ at least\ $1/\gamma^2$ independently, $X_e=1$. Thus,
    \[\E\left[\sum_{e\in M}X_e\right]\geq \frac{|M|}{\gamma^2}.\]
    On the other hand,
    \begin{align*}
        \E\left[\sum_{e\in M}X_e\right]
        &\leq \Pr\left[\sum_{e\in M} X_e\geq \frac{|M|}{2\gamma^2}\right]\cdot |M|+\left(1-\Pr\left[\sum_{e\in M} X_e\geq \frac{|M|}{2\gamma^2}\right]\right)\cdot \frac{|M|}{2\gamma^2}\\
        &\leq \Pr\left[\sum_{e\in M} X_e\geq \frac{|M|}{2\gamma^2}\right]\cdot |M|+\frac{|M|}{2\gamma^2}.
    \end{align*}
    Therefore,
    \[\Pr\left[\sum_{e\in M} X_e\geq \frac{|M|}{2\gamma^2}\right]\geq \frac{1}{2\gamma^2},\]
    and by repeating $k=O(\gamma^2\log n)$ times independently, there is a sampled subgraph $\B_G[S^1\cup S^2]$ satisfies $\mu(\B_G[S^1\cup S^2])\geq \frac{M}{2\gamma^2}$ with probability $1-\frac{1}{n^{O(1)}}$.

    By the guarantees of the oracle $\O$, the matching $M^\prime/\Omega$ then satisfies:
    \[|M^\prime/\Omega|=|M^\prime|\geq \alpha\cdot \frac{|M|}{2\gamma^2}-\beta\cdot n\geq \frac{\alpha}{2\gamma^2}\cdot \mu(\B_{G/\Omega}[T^1\cup T^2])-(\beta\gamma)\cdot |V_{G/\Omega}|,\]
    where $V_{G/\Omega}$ is the vertex set of $G/\Omega$.

\end{proof}

We now wrap up and give the final lemma for simulating induced matching 
oracles on $\B_{\Gtight/\Omega}$.
\begin{lemma}\label{lemma:reduction:induced-tight-to-eps-extension}
    Given an $(\alpha,\beta)$-approximate induced matching oracle $\O$ initialized on $G^\eps$ with query time $\U$.
    Consider the set $\{M,\Omega,\by,\bz\}$ at the beginning of the $t$-th round, and the corresponding contracted tight subgraph $\Gtight/\Omega$ s.t.\ the largest blossom in $\Omega$ has size at most $\gamma$, an $(\Omega(\alpha/\gamma^2),O(\beta\gamma/\eps))$-approximate induced matching query on $\B_{\Gtight/\Omega}$ can be answered (w.h.p.\ against an adaptive adversary) in time
    \[O(\gamma^2\log n)\cdot \U.\]
\end{lemma}
\begin{proof}
    As suggested by~\cref{lemma:tight-equivalence}, we will build an approximate induced matching oracle on $\B_{\Gtight^\prime}$, and then apply~\cref{lemma:reduction:contracted-to-uncontracted} to get the oracle on $\B_{\Gtight^\prime/\Omega}=\B_{\Gtight/\Omega}$.
    
    Recall that the edge set of $\B_{\Gtight^\prime}$ is the union of $E_1=\{\{u^1,v^2\}:\bw(\{u,v\})\leq y_u+y_v\leq \bw(e)+2\eps\cdot W\}$ and $E_2=\{\{u^1,v^2\}:\{u,v\}\in M\}$.
    During~\cref{alg:outer-loop}, $\by$-values are non-negative multiples of $\eps\cdot W$. For an induced matching query $S^1\cup S^2$ on $\B_{\Gtight^\prime}$ where $S^1$ is in the left part and $S^2$ is in the right part, define $S^1_\eps=\{u^1_{y_u}:u^1\in S^1\}$ and $S^2_\eps=\{v^2_{y_v}:v^2\in S^2\}$ (here $y_u$ and $y_v$ denote the dual $\by$-values of $u$ and $v$). Then, it is easy to see that $\B_{\Gtight^\prime}[S^1\cup S^2]\cap E_1$ is equivalent to $G^\eps[S^1_\eps\cup S^2_\eps]$. We use $\O$ to compute an $(\alpha,\beta)$-approximate induced matching $M_1\subseteq G^\eps[S^1_\eps\cup S^2_\eps]$, thus, $|M_1|\geq \alpha\cdot \mu(\B_{\Gtight^\prime}[S^1\cup S^2]\cap E_1)-\beta\cdot O(n/\eps)$. Note that $E_2$ is itself a matching; the final output of our oracle is whichever of $M_1$ and $E_2[S^1\cup S^2]$ is bigger, which we denote with matching $M'$.  
    Since $|M^\prime|\geq \alpha\cdot \mu(\B_{\Gtight^\prime}[S^1\cup S^2]\cap E_1)-\beta\cdot O(n/\eps)$ and also $|M^\prime|\geq \mu(\B_{\Gtight^\prime}[S^1\cup S^2]\cap E_2)$ and $\mu(\B_{\Gtight^\prime}[S^1\cup S^2])\leq \mu(\B_{\Gtight^\prime}[S^1\cup S^2]\cap E_1)+\mu(\B_{\Gtight^\prime}[S^1\cup S^2]\cap E_2)$, we have
    $|M^\prime|\geq \Omega(\alpha)\cdot \mu(\B_{\Gtight^\prime}[S^1\cup S^2])-O(\beta/\eps)\cdot O(n)$.

    Now we've built an $(\Omega(\alpha),O(\beta/\eps))$-approximate induced matching oracle on $\B_{\Gtight^\prime}$ with query time $\U$. By~\cref{lemma:reduction:contracted-to-uncontracted}, this gives us an $(\Omega(\alpha/\gamma^2),O(\beta\gamma/\eps))$-approximate induced matching oracle on $\B_{\Gtight^\prime/\Omega}$ with query time $O(\gamma^2\log n)\cdot \U$.
\end{proof}

\subsection{Implementing Adjacency Matrix Oracle on Contracted Tight Subgraphs}\label{sec:implement-adjacenct-matrix}

Based on~\cref{lemma:tight-equivalence}, it also suffices to implement an adjacency matrix oracle on $\B_{\Gtight^\prime/\Omega}$.
\begin{lemma}\label{lemma:implement-adjacency-matrix-oracle}
   Consider the set $\{M,\Omega,\by,\bz\}$ at the beginning of the $t$-th round, and the corresponding contracted tight subgraph $\Gtight/\Omega$ s.t.\ the largest blossom in $\Omega$ has size at most $\gamma$, there is an adjacency matrix oracle on $\B_{\Gtight^\prime/\Omega}$ with an query time $O(\gamma^2)$.
\end{lemma}

\begin{proof}
    For any query $\{B_u,B_v\}$ in $\Gtight/\Omega$, by~\cref{lemma:tight-equivalence}, it suffices to answer the query in $\B_{\Gtight^\prime/\Omega}$. We go through every possible edge $\{u,v\}\in G$ s.t.\ $u\in B_u$ and $v\in B_v$ and check whether $\bw(\{u,v\})\leq y_u+y_v\leq \bw(\{u,v\})+2\eps\cdot W$ or $\{u,v\}\in M$. Since $|B_u|,|B_v|\leq \gamma$, the total number of edges that we need to check is $O(\gamma^2)$.
\end{proof}

\subsection{Putting Everything Together}
Combining~\cref{lemma:structural:weighted:weaker-version,lemma:reduction:induced-tight-to-eps-extension,lemma:implement-adjacency-matrix-oracle}, we are now ready to prove~\cref{lemma:structural:weighted}.

\begin{proof}[Proof of~\cref{lemma:structural:weighted}]
By~\cref{lemma:framework:blossom-size}, $\gamma=O(\eps^{-2})$.
By~\cref{lemma:reduction:induced-tight-to-eps-extension}, an $(\Omega(\alpha\cdot\eps^4),O(\alpha\cdot\eps^{82}))$-approximate induced matching query on $\B_{\Gtight/\Omega}$ can be answered in time $O(\eps^{-4}\log n)\cdot \U(\alpha,O(\alpha\cdot \eps^{85}))$ assuming the induced matching oracle is initialized on $G^\eps$. By~\cref{lemma:implement-adjacency-matrix-oracle}, an adjacency matrix query can be answered in $O(\eps^{-4})$ time, and its total cost will be dominated. Thus the total runtime based on~\cref{lemma:structural:weighted:weaker-version} is
\[O(\alpha^{-1}\eps^{-123}\log(\alpha^{-1}\eps^{-1})\log^2n)\cdot \U_\O(\alpha,O(\alpha\cdot\eps^{85})).\]
\end{proof}

\section{Improved Dynamic Weighted Matching Algorithms}
In the dynamic matching problem, the underlying graph undergoes edge updates. Denote $G_0$ as the initial graph, and $G_i$ as the graph after the $i$-th update. The goal is to maintain a sequence of matching $M_i$ s.t.\ $M_i$ is a $(1-O(\eps),0)$-approximate MWM of $G_i$. \cref{lemma:add-apx-to-multi-apx} reduces the task of dynamically maintaining a $(1-O(\eps),0)$-approximate MWM to dynamically maintaining a $(1,O(\eps)\cdot W\cdot n)$-approximate MWM. Let $t_i=i\cdot\eps\cdot n$. Following~\cref{lemma:lazy-update}, it suffices to find the matching on $G_{t_i}$ for all $t_i$, and the computational cost can be amortized over $\eps\cdot n$ updates.

\begin{lemma}[Lazy Update]\label{lemma:lazy-update}
    Suppose $M_{t_i}$ is a $(1,c\cdot \eps\cdot W\cdot n)$-approximate MWM on $G_{t_i}$. For any $t$ such that\ $t_i<t<t_{i+1}$, $M_{t_i}\cap G_t$ is a $(1,(c+2)\cdot \eps\cdot W\cdot n)$-approximate MWM on $G_t$.
\end{lemma}
\begin{proof}
    Denote $U$ as the edge updates between $t_i$ and $t$, we have $\bw(U)\leq \eps\cdot W\cdot n$. The weight of the matching is at least
    \[\bw(M_{t_i}\cap G_t)\geq \bw(M_{t_i})-\bw(U)\geq \mu_{\bw}(G_{t_i})-(c+1)\cdot\eps\cdot W\cdot n,\]
    while the MWM of the graph has weight at most
    \[\mu_{\bw}(G_t)\leq \mu_{\bw}(G_{t_i})+\bw(U)\leq \mu_{\bw}(G_{t_i})+\eps\cdot W\cdot n.\]
    Thus
    \[\bw(M_{t_i}\cap G_t)\geq \bw(G_t)-(c+2)\cdot \eps\cdot W\cdot n.\]
\end{proof}
\cref{lemma:structural:weighted} then reduces the task of finding a $(1,O(\eps)\cdot W\cdot n)$-approximate MWM on $G_{t_i}$ to answering $(\alpha,\beta)$-approximate induced matching queries on $G_{t_i}^\eps$. In this section, we solve the induced matching queries in both the online and offline dynamic settings; recall that online (the standard setting) means that the algorithm learns about the updates one at a time, whereas offline means that the the algorithm is told the entire sequence of updates in advance. 

\subsection{Online Dynamic Implementation of the Induced Matching Oracle}\label{sec:dynamic-implementation}
A high-level idea on how to answer an induced matching query on an $n$-node graph using an online dynamic matching algorithm with $O(n)$ updates is based on the following observation.

\begin{lemma}\label{lemma:auxiliary-transformation}
    For any $n$-node graph $G=(V,E)$, denote $G^\prime=(V^\prime,E)$ based on $G$ where $V^\prime=V\cup \{u^\prime:u\in V\}$, i.e., add an auxiliary vertex $u^\prime$ for each vertex $u\in V$. For an induced matching query $S\subseteq V$ in $G$, consider the graph $G^\prime_S=(V^\prime,E_S^\prime)$ based on $G^\prime$ where $E_S^\prime=E\cup\{\{u,u^\prime\}:u\in V\setminus S\}$, i.e., add an auxiliary edge $\{u,u^\prime\}$ for any $u\in V\setminus S$. Then for any $(1-\eps,0)$-approximate MCM $M$ of $G_S^\prime$, $M[S]$ is a $(1,\eps\cdot n)$-approximate MCM of $G[S]$.
\end{lemma}

\begin{proof} First we show $\mu(G_S^\prime)\geq n-|S|+\mu(G[S])$ through construction: pick an MCM of $G[S]$ and add $\{u^\prime,u\}$ for all $u\in V\setminus S$. It is clear that edges are vertex-disjoint with size $n-|S|+\mu(G[S])$.

Now suppose we have a $(1-\eps,0)$-approximate MCM $M$ of $G_S^\prime$. The edges in $M$ but not in $G[S]$ must be connected to a vertex in $V\setminus S$. Since $M$ is a matching, the number of such edges is upper bounded by $|V\setminus S|=n-|S|$. Thus $|M[S]|\geq (1-\eps)\cdot \mu(G_S^\prime)-(n-|S|)\geq (1-\eps)\cdot (n-|S|+\mu(G[S]))-(n-|S|)=\mu(G[S])-\eps\cdot (n-|S|+\mu(G[S]))$.
Since $\mu(G[S])\leq |S|/2\leq |S|$, we have
$|M[S]|\geq \mu(G[S])-\eps\cdot n$.
\end{proof}

In other words, for a fixed graph $G$, to answer an induced matching query $S$ given a dynamic matching algorithm, it suffices to send $O(n)$ edge updates to $G^\prime$ 
to add the corresponding edges $\{u,u^\prime\}$ to form $G^\prime_S$, and then use the returned matching. To make sure that incremental/decremental matching algorithms with worst-case update time can also complete the job, we need different types of initialization of the auxiliary graph.

\begin{definition}[Incremental Auxiliary Graph]
\label{def:inc-aux}
    For any graph $G=(V,E)$, its incremental auxiliary graph is defined to be $G^\prime$, i.e., add to $G$ an auxiliary vertex $u^\prime$ for every vertex $u\in V$.
\end{definition}

\begin{definition}[Decremental Auxiliary Graph]
\label{def:dec-aux}
    For any graph $G=(V,E)$, its decremental auxiliary graph is defined to be $G^\prime_{\emptyset}$, i.e., add to $G$ an auxiliary vertex $u^\prime$ and an auxiliary edge $\{u,u^\prime\}$ for every vertex $u\in V$.
\end{definition}

Then, for an incremental (decremental) algorithm, given the query $S$, the algorithm can first insert (delete) edges into (from) the auxiliary graph to arrive at $G^\prime_S$, then return the matching in $G^\prime_S$, and finally roll back all those updates to return to the auxiliary graph.
For a fully dynamic algorithm, it can choose to maintain either type of auxiliary graphs. To avoid ambiguity, when we say the auxiliary graph for a fully dynamic algorithm, we refer to the incremental auxiliary graph.

\subsection{Online Dynamic Reduction: from General MWM to Bipartite MCM}
We are now ready to give the full reduction from dynamic $(1-O(\eps),0)$-approximate MWM algorithms to dynamic $(1-O(\eps),0)$-approximate MCM algorithms. We start with the reduction for fully dynamic algorithms.

\begin{theorem}\label{thm:reduction:fully dynamic}
    Say we are given a $(1-\eps,0)$-approximate fully dynamic MCM algorithm $\A$ that, on an input $n$-vertex $m$-edge bipartite graph, has initialization time $\I(n,m,\eps)$, amortized
    update time $\U(n,m,\eps)$. Then, there is a $(1-O(\eps\log(\eps^{-1})),0)$-approximate fully dynamic MWM algorithm (w.h.p.\ against an adaptive adversary), on an input $n$-vertex $m$-edge $W$-aspect ratio general graph (not necessarily bipartite), that has initialization time
    \[O(\eps^{-6}\log^2 n\log(\eps^{-1}))\cdot (\I(O(n/\eps),O(m/\eps),O(\eps^{510}))+n)+O(m\eps^{-1}\log(\eps^{-1}))\]
    amortized update time
    \[\poly(\log(\eps^{-1}))\cdot O(\eps^{-756}\log^4n)\cdot \U(O(n/\eps),O(m/\eps),O(\eps^{510})).\]
\end{theorem}

\begin{proof} We first build a $(1,\delta)$-induced matching oracle of $G_{t_i}^\eps$, for each $t_i=i\cdot\eps\cdot n$. We initialize $\A$ on the auxiliary graph of $G_0^\eps$ in time $\I(O(n/\eps),O(m/\eps),\delta)$ (\Cref{def:inc-aux}). For each edge update in the graph, we send the corresponding update to $\A$, thus, the underlying graph of $\A$ is always the auxiliary graph of $G_i^\eps$ after the $i$-th update. For each induced matching query $S$, we insert corresponding auxiliary edges to arrive at the graph $(G_{t_i}^\eps)^\prime_{S}$ (defined in~\cref{lemma:auxiliary-transformation}). The number of insertions is at most the number of vertices in $G_{t_i}^\eps$, which is $O(n/\eps)$.
According to~\cref{lemma:auxiliary-transformation}, the matching maintained by $\A$ answers a $(1,\delta)$-approximate induced matching query. The query time of the induced matching oracle is then $O(n/\eps)\cdot \U(O(n/\eps),O(m/\eps),\delta)$.

Then~\cref{lemma:structural:weighted} sets $\delta=O(\eps^{85})$ and computes a $(1,\eps\cdot W\cdot n)$-approximate MWM on $G_{t_i}$ in time $O(\eps^{-124}\log(\eps^{-1})\cdot n\log^2n)\cdot \U(O(n/\eps),O(m/\eps),O(\eps^{85}))$. According to~\cref{lemma:lazy-update}, we now maintain a $(1,O(\eps)\cdot W\cdot n)$-approximate MWM on all $G_i$ with amortized update time $O(\eps^{-125}\log(\eps^{-1})\log^2n)\cdot \U(O(n/\eps),O(m/\eps),O(\eps^{85}))$.

Finally, using~\cref{lemma:add-apx-to-multi-apx}, there is a $(1-O(\eps\log(\eps^{-1})),0)$-approximate fully dynamic MWM algorithm with initialization time
\[O(\log^2 n\eps^{-6}\log(\eps^{-1}))\cdot (\I(O(n/\eps),O(m/\eps),O(\eps^{510}))+n)+O(m\eps^{-1}\log(\eps^{-1})),\]
and amortized update time
\[\poly(\log(\eps^{-1}))\cdot O(\log^4n\cdot \eps^{-756})\cdot \U(O(n/\eps),O(m/\eps),O(\eps^{510})).\]
\end{proof}

We also state the reduction for partially dynamic algorithms. But here we need the black-box algorithm to have worst-case update time because we will send $O(n)$ edge updates to the algorithm for induced matching queries, and those updates will roll back, as mentioned in~\cref{sec:dynamic-implementation}. 
The proof remains the same as the fully dynamic case.

\begin{theorem}\label{thm:reduction:partially dynamic}
    Given a $(1-\eps,0)$-approximate incremental/decremental MCM algorithm $\A$ that, on an input $n$-vertex $m$-edge bipartite graph, has initialization time $\I(n,m,\eps)$, \textbf{worst-case} update time $\U(n,m,\eps)$, there is a $(1-O(\eps\log(\eps^{-1})),0)$-approximate incremental/decremental MWM algorithm (w.h.p.\ against an adaptive adversary), on an input $n$-vertex $m$-edge $W$-aspect ratio general graph (not necessarily bipartite), that has initialization time
    \[O(\eps^{-6}\log^2 n\log(\eps^{-1}))\cdot (\I(O(n/\eps),O(m/\eps),O(\eps^{510}))+n)+O(m\eps^{-1}\log(\eps^{-1}))\]
    amortized update time
    \[\poly(\log(\eps^{-1}))\cdot O(\eps^{-756}\log^4n)\cdot \U(O(n/\eps),O(m/\eps),O(\eps^{510})).\]
\end{theorem}

\subsection{Improved Offline Dynamic Matching Algorithm}
In the offline dynamic setting, we can use fast matrix multiplication to answer induced matching queries, as suggested by the following lemma in~\cite{Liu24}.

\begin{lemma}[{\cite[Lemma 2.11]{Liu24}}]\label{lemma:offline-implementation}
    Let $H_1,\cdots,H_k$ be graphs such that $H_i$ and $H_1$ differ in at most $\Gamma$ edges. Let $S_i\subseteq V$ for $i\in[k]$. There is a randomized algorithm that returns a $(2,0)$-approximate matching on each $H_i[S_i]$ for $i\in[k]$ w.h.p.\ in total time
    \[\otilde(k\Gamma+n^2k/D+D\cdot T(n,n/D,k)),\]
    for any $D>0$. $T(r,s,t)$ is the runtime needed to multiply a $r\times s$ matrix by an $s\times t$ matrix.
\end{lemma}

Based on~\cref{lemma:offline-implementation}, we build an offline dynamic $(1-\eps,0)$-approximate MWM algorithm.

\OfflineDynamic*

\begin{proof}
    We first implement $(2,0)$-approximate induced matching oracles on $\{G^\eps_{t_i}\}$. Pick $k=(n/\eps)^x$ and $D=(n/\eps)^y$. Apply~\cref{lemma:offline-implementation} on $G_{t_{i+1}}^\eps,G_{t_{i+2}}^\eps,\cdots,G_{t_{i+k}}^\eps$ for an arbitrary $i$ with $\Gamma=k\cdot n$. A $(2,0)$-approximate induced matching query is answered in amortized (over $\eps\cdot k\cdot n$ updates) time
    \[\otilde(\poly(\eps^{-1}))\cdot (n^{x}+n^{1-y}+n^{y-1-x}\cdot T((n/\eps),(n/\eps)^{1-y},(n/\eps)^x)).\]
    By picking $x=0.579$ and $y=0.421$, and combining with~\cref{lemma:structural:weighted,lemma:add-apx-to-multi-apx}, there is a randomized $(1-O(\eps),0)$-approximate offline dynamic MWM algorithm the amortized update time
    \[O(\poly(\eps^{-1})\cdot n^{0.58}).\]
\end{proof}

\section{Reduction from Fully Dynamic to Decremental Approximate MCM}\label{sec:decremental-hardness}
\cite{Liu24} reduces the fully dynamic matching problem to the fully dynamic approximate OMv problem. In this section, we first reduce the fully dynamic approximate OMv problem to the decremental approximate OMv problem (\cref{lemma:decremental2full:omv}). Then we reduce the decremental approximate OMv problem to the decremental matching problem with worst-case update time (\cref{lemma:decrementalmatching2omv}). At the end of this section, we will show that a sublinear worst-case update time decremental matching algorithm can be transformed into a sublinear update time fully dynamic matching algorithm (\cref{thm:matching:decremental2full}). Below, we give the definitions of the approximate OMv problem, the fully dynamic approximate OMv problem, and the decremental approximate OMv problem.
\begin{definition}[Approximate OMv,~{\cite[Definition 1.3]{Liu24}}]
In the $(1-\gamma)$-approximate OMv problem, an algorithm is given a Boolean matrix $M\in\{0,1\}^{n\times n}$. After preprocessing, the algorithm receives an online sequence of query vectors $v^{(1)},v^{(2)},\dots,v^{(n)}\in\{0,1\}^n$. After receiving $v^{(i)}$, the algorithm must respond with a vector $w^{(i)}\in\{0,1\}^n$ such that $d(Mv^{(i)},w^{(i)})\leq \gamma n$, where $Mv^{(i)}$ is the Boolean matrix product, and $d(\cdot,\cdot)$ is the Hamming distance.
\end{definition}

\begin{definition}[Fully Dynamic Approximate OMv,~{\cite[Definition 1.4]{Liu24}}]
In the $(1-\gamma)$-approximate fully dynamic OMv problem, an algorithm is given a matrix $M\in\{0,1\}^{n\times n}$, initially 0. Then, it responds to the following:
\begin{itemize}
    \item $\mathtt{Update(i,j,b)}$: set $M_{ij}=b$.
    \item $\mathtt{Query(v)}$: output a vector $w\in\{0,1\}^n$ with $d(Mv,w)\leq \gamma n$.
\end{itemize}
\end{definition}

\begin{definition}[Decremental Approximate OMv]
In the $(1-\gamma)$-approximate decremental OMv problem, an algorithm is given a matrix $M\in\{0,1\}^{n\times n}$. Then, after initialization, it responds to the following:
\begin{itemize}
    \item $\mathtt{Delete(i,j)}$: set $M_{ij}=0$.
    \item $\mathtt{Query(v)}$: output a vector $w\in\{0,1\}^n$ with $d(Mv,w)\leq \gamma n$.
\end{itemize}
\end{definition}

\cref{lemma:decremental2full:omv} reduces the fully dynamic OMv problem to the decremental OMv problem. The idea is to only process the deletion with the algorithm and restore the insertion independently with brute force, and repeatedly rebuild the matrix after a certain number of updates.

\begin{lemma}\label{lemma:decremental2full:omv}
Given an algorithm $\A$ that solves $(1-\gamma)$-approximate decremental OMv problem with initialization time $\I(n,\gamma)$, update time $\U(n,\gamma)$ and query time $\Q(n,\gamma)$, then there is a $(1-\gamma)$-approximate fully dynamic OMv algorithm with amortized update time
\[O(\I(n,\gamma)/\Q(n,\gamma)+\U(n,\gamma)),\]
and query time
\[O(\Q(n,\gamma)).\]
\end{lemma}

\begin{proof}
    We will build the $(1-\gamma)$-approximate fully dynamic OMv algorithm by recursively rebuilding. We do the following for every $\Q(n,\gamma)$ update.
    We initialize $\A$ according to the current $M$ in $\I(n,\gamma)$ time. Set $M^\prime=M$ and $S=\emptyset$ for future use. For the next $\Q(n,\gamma)$ queries $\mathtt{Update(i,j,b)}$, if $M^\prime_{ij}=1$ and $b=0$, $\A.\mathtt{Delete(i,j)}$, and set both $M^\prime_{ij}$ and $M_{ij}$ to 0. Otherwise, set $M_{ij}$ to $b$ and $S\gets S\cup\{\{i,j\}\}$. For any query $v$, set $w=\A.\mathtt{Query(v)}$, and then scan through all edges $\{i,j\}$ in $S$, if $v_j=1$ and $M_{ij}=1$ then set $w_i=1$. Then output $w$ as the answer.

    The update time is  $O(\I(n,\gamma)/\Q(n,\gamma)+\U(n,\gamma))$ because we rebuild every $\Q(n,\gamma)$ update and call $\A.\mathtt{Delete}$ at most once for each update. The query time is $O(\Q(n,\gamma))$ because we call $\A.\mathtt{Query}$ at most once per update, and the size of the set $S$ is also bounded by $\Q(n,\gamma)$.

    Now we prove that the above algorithm solves the fully dynamic approximate OMv. It's easy to see that our algorithm keeps the following properties: 1) $M_{ij}\geq M^\prime_{ij}$ and 2) If $M_{ij}>M^\prime_{ij}$ then $\{i,j\}\in S$. By definition of the decremental OMv algorithm,
    $d(M^\prime v,\A.\mathtt{Query(v)})\leq \gamma n$. It remains to argue that for any index $i$ such that $(M^\prime v)_i=(\A.\mathtt{Query(v)})_i$, we will have $(Mv)_i=w_i$.
    \begin{enumerate}
        \item If $(M^\prime v)=1$, then $(Mv)_i=1$ since $M_{ij}\geq M^\prime_{ij}$. Also, in the algorithm described above, $w_i\geq (\A.\mathtt{Query(v)})_i=1$ thus $(Mv)_i=w_i$.
        \item If $(M^\prime v)_i=0$ and $(Mv)_i=0$, we have for all $\{i,j\}\in S$, $v_j$ and $M_{ij}$ can't both be 1. Thus $w_i=0$.
        \item If $(M^\prime v)_i=0$ and $(Mv)_i=1$, there exists a $j$ such that $M_{ij}=1, v_j=1$ and $M^\prime_{ij}=0$. By the property above we have $\{i,j\}\in S$, thus $w_i=1$.
    \end{enumerate}

\end{proof}

\begin{algorithm}[!ht]
  \caption{\textsc{Decremental Approximate OMv}} \label{alg:decremental-approximate-omv}
  
  \SetEndCharOfAlgoLine{}
  \SetKwInput{KwData}{Input}
  \SetKwComment{Comment}{/* }{ */}
  \SetKwProg{KwProc}{function}{}{}
  \SetKwFunction{Initialize}{Initialize}
  \SetKwFunction{Delete}{Delete}
  \SetKwFunction{Query}{Query}
  \SetKwFunction{SubsetPreparation}{SubsetPreparation}
  \SetKwFunction{Undo}{Undo}

  \KwData{Matrix $M\in\{0,1\}^{n\times n}$, $(1-\eps
  )$-approximate decremental matching algorithm $\A$.}

  \KwProc{\Initialize{$M\in\{0,1\}^{n\times n}$}} {
    Denote $B(L\cup R,E)$ as a bipartite graph such that $L=\{u_1^L,\dots,u_n^L,\tilde u_1^L,\dots,\tilde u_n^L\}, R=\{u_1^R,\dots,u_n^R,\tilde u_1^R,\dots,\tilde u_n^R\}$ and $E=\{\{u_i^L,\tilde u_i^L\},\{u_i^R,\tilde u_i^R\}:1\leq i\leq n\}\cup \{\{u_i^L,u_j^R\}:M_{ij}=1\}$.\;
    $\A$.\Initialize{$B$}.\;
  }
  
  \KwProc{\Delete{$i,j$}} {
    $\A$.\Delete{$\{u^L_i,v^R_j\}$}.\;
    $E\gets E\setminus \{u^L_i,v^R_j\}$.\;
    $M_{ij}\gets 0$.\;
  }

  \KwProc{\SubsetPreparation{$S,T$}} {
    $C\gets \emptyset$.\;
    During the following operations, we will use $C$ to record every change made in the state and memory of $\A$.\;
    \lFor{$s\in S$}{$\A$.\Delete{$\{u_s^L,\tilde u_s^L\}$}.}
    \lFor{$t\in T$}{$\A$.\Delete{$\{u_t^R,\tilde u_t^R\}$}.}
  }
  
  \KwProc{\Undo{}}{
    In the reverse order of $C$, undo the recorded changes.
  }

  \KwProc{\Query{$\bv$}}  {
    $\bw\gets 0^n$\;
    $S\gets \{1,2,\dots,n\},T\gets\{i:v_i=1\}$.\;
    \For{$t_L=1,\dots,T_L\defeq 10\eps^{1/4}n^2\log n$}{
        Pick $s\in S,t\in T$ uniformly at random.\;
        \lIf{$M_{st}=1$}{$w_s\gets 1,S\gets S\setminus \{s\}$.}
    }
    \For{$t_R=1,\dots,T_R\defeq 10\eps^{1/2} n^2\log n$}{
        Pick $s\in S, t\in T$ uniformly at random.\;
        \If{$M_{st}=1$}{
            \lFor{all $s^\prime\in S$ such that $M_{s^\prime t}=1$}{
                $w_{s^\prime}\gets 1,S\gets S\setminus\{s^\prime\}$.
            }
            $T\gets T\setminus \{t\}$.\;
        }
    }
    \SubsetPreparation{$S,T$}.\;
    \While{$\exists$ at least $\eps n$ edges between $\{u_s^L:s\in S\}$ and $\{u_t^R:t\in T\}$ in $\A.\mathtt{Matching}$}{
        \For{$s\in S$ such that $u_s^L~\text{is matched to some vertex in }\{u_t^R:t\in T\}\text{ in }\A.\mathtt{Matching}$}{
            $w_s\gets 1,S\gets S\setminus \{s\}$.\;
            $\A$.\Delete{$\{u_s^L,\tilde u_s^L\}$}.\;
        }
    }
    \Undo{}.\;
    \textbf{return} $\bw$.\;
  }

\end{algorithm}

\cref{lemma:decrementalmatching2omv} reduces the decremental OMv problem to the decremental matching problem with worst-case update time through~\cref{alg:decremental-approximate-omv}.

\begin{lemma}\label{lemma:decrementalmatching2omv}
    Given an algorithm that solves the $(1-\eps)$-approximate decremental matching problem on bipartite graphs with initialization time $\I(n,m,\eps)$, \textbf{worst-case} update time $\U(n,m,\eps
    )$, then~\cref{alg:decremental-approximate-omv} solves the $(1-O(\sqrt{\eps}))$-approximate decremental OMv problem w.h.p.\ for a single query with initialization time
    \[\I(n,O(nnz(M)+n),\eps),\]
    worst-case update time
    \[\U(n,O(nnz(M)+n),\eps),\]
    and worst-case query time
    \[\otilde(n^2\sqrt{\eps}+n\cdot\U(n,O(nnz(M)+n),\eps)),\]
    where $M\in\{0,1\}^{n\times n}$ is the initial matrix in the decremental approximate OMv problem.
\end{lemma}

\begin{proof}
    We first prove that the algorithm answers with probability $1-O(n^{-10})$ each query $\bv$ with $\bw$ such that $d(M\bv,\bw)\leq O(\sqrt{\eps} n)$. Denote $G^\prime$ as the subgraph of $B$ induced by the vertex set $\{u_i^L,u_j^R:1\leq i\leq n,v_j=1\}$, and $\deg(u)$ as the degree of vertex $u$ in this subgraph. $(M\bv)_i$ denotes whether vertex $u_i^L$ has an incident edge in the subgraph. For any vertex $u_i^L$ with $\deg(u_i^L)\geq \eps^{-1/4}$, some edge $(u_i^L, u_j^R)$ in the subgraph will be picked with probability at least $1-(1-\eps^{-1/4}/n^2)^{T_L}\geq 1-n^{-10}$. Thus after the first for-loop, $S$ contains only vertices with degree at most $\eps^{-1/4}$. Similarly, after the second for-loop, $T$ contains only vertices with degree at most $\eps^{-1/2}$. When the while-loop doesn't hold,~\cref{lemma:auxiliary-transformation} shows that MCM in the matching size is at most $2\eps n$. Since the maximum degree in the subgraph induced by $\{u_s^L,u_t^R:s\in S,t\in T\}$ is at most $\eps^{-1/2}$, the number of edges in the subgraph is at most $2\eps^{1/2}n$. Meaning $d(M\bv,\bw)\leq 2\eps^{1/2}n$.

    Now we consider the runtime. The initialization time and update time follow directly by definition. For the query, the first for-loop costs $\otilde(n^2\eps^{1/4})$. In the second for-loop, for those vertex $u_t^R$ with $\deg(u_t^R)\geq \eps^{-1/4}$, there will be at most $n\eps^{1/4}$ of them and the total cost is $n^2\eps^{1/4}$. For those vertices with at most $\eps^{-1/4}$ the number of corresponding edges in the induced subgraph is bounded by $\eps^{-1/4}|T|$. With probability $1-n^{-10}$, the algorithm scan through $S$ with at most $\otilde(\eps^{-1/4}|T|\cdot T_R/(|S|\times |T|))$ times. Thus the second for-loop costs $\otilde(n^2\eps^{1/4})$. For the remaining operations in $\A$, since $\A$ has worst-case update time $\U(n,O(nnz(M)+n),\eps)$, \SubsetPreparation{}, the edge deletions and and \Undo{} costs $O(n\cdot \U(n,O(nnz(M)+n),\eps))$. The cost of reading the edges in the matching has the same time bound because the worst-case update time upper bounds the worst-case recourse which is the difference in the edge set between updates.
\end{proof}

Now we are ready to state the full reduction from fully dynamic matching to decremental matching with worst-case update time.

\PartialHardness*

\begin{theorem}[{\cite[Theorem 2]{Liu24}}]\label{thm:equivalence:matchingvsOmv}
    There is an algorithm solving dynamic $(1-\gamma)$-approximate OMv with $\gamma=n^{-\delta}$ with amortized $n^{1-\delta}$ for $\mathtt{Update}$ and $n^{2-\delta}$ for $\mathtt{Query}$, for some $\delta>0$ against adaptive adversaries, if and only if there is a randomized algorithm that maintains a $(1-\eps)$-approximate dynamic matching with amortized time $n^{1-c}\eps^{-C}$, for some $c,C>0$ against adaptive adversaries.
\end{theorem}

\begin{proof}
    By~\cref{lemma:decrementalmatching2omv}, and picking $\gamma=n^{-\alpha}$ for some $\alpha>0$ small enough, there is a randomized algorithm that solves the $(1-n^{-\theta/2})$-approximate decremental OMv problem with initialization time $\poly(1/\gamma)\cdot n^{3-\delta}=O(n^{3-\theta})$, worst-case update time $\poly(1/\gamma)\cdot n^{1-\delta/2}=O(n^{1-\theta/2})$ and worst-case query time $\otilde(n^2\sqrt{\gamma}+\poly(1/\gamma)\cdot n^{2-\delta/2})=O(n^{2-\theta/2})$ for some constant $0<\theta\leq 1$.

    By~\cref{lemma:decremental2full:omv}, there is a $(1-n^{-\theta/2})$-approximate fully dynamic OMv algorithm with amortized update time $O(n^{1-\theta/2})$ and query time $O(n^{2-\theta/2})$.

    Combined with~\cref{thm:equivalence:matchingvsOmv}, we finish the proof.
\end{proof}

\section*{Acknowledgements}
Thank you to Aaron Sidford, Ta-Wei Tu, Aditi Dudeja, Slobodan Mitrović and Wen-Horng Sheu for the useful discussions of the project. Thank you to the anonymous reviewers for their helpful feedback.

\bibliographystyle{alpha}
\bibliography{reference}

\newcommand{\etalchar}[1]{$^{#1}$}
\begin{thebibliography}{GKMS19}

\bibitem[ABD22]{AssadiBD22}
Sepehr Assadi, Aaron Bernstein, and Aditi Dudeja.
\newblock Decremental matching in general graphs.
\newblock In {\em Proc. 49th Int. Colloquium on Automata, Languages, and Programming}, volume 229 of {\em LIPIcs}, pages 11:1--11:19, 2022.
\newblock Available at \url{https://arxiv.org/abs/2207.00927}.

\bibitem[AKK25]{AssadiKK25}
Sepehr Assadi, Sanjeev Khanna, and Peter Kiss.
\newblock Improved bounds for fully dynamic matching via ordered ruzsa-szemer{\'{e}}di graphs.
\newblock In {\em Proceedings of the 2025 Annual {ACM-SIAM} Symposium on Discrete Algorithms, {SODA} 2025}, pages 2971--2990. {SIAM}, 2025.
\newblock Available at \url{https://arxiv.org/abs/2406.13573}.

\bibitem[BCD{\etalchar{+}}25]{BernsteinCDLST24}
Aaron Bernstein, Jiale Chen, Aditi Dudeja, Zachary Langley, Aaron Sidford, and Ta{-}Wei Tu.
\newblock Matching composition and efficient weight reduction in dynamic matching.
\newblock In {\em Proceedings of the 2025 Annual {ACM-SIAM} Symposium on Discrete Algorithms, {SODA} 2025}, pages 2991--3028. {SIAM}, 2025.
\newblock Available at \url{https://arxiv.org/abs/2410.18936}.

\bibitem[BDL21]{BernsteinDL21}
Aaron Bernstein, Aditi Dudeja, and Zachary Langley.
\newblock A framework for dynamic matching in weighted graphs.
\newblock In {\em Proceedings of the 53rd Annual {ACM} Symposium on Theory of Computing, {STOC} 2021}, pages 668--681. {ACM}, 2021.

\bibitem[BG24]{BehnezhadG24}
Soheil Behnezhad and Alma Ghafari.
\newblock Fully dynamic matching and ordered ruzsa-szemer{\'{e}}di graphs.
\newblock In {\em 65th {IEEE} Annual Symposium on Foundations of Computer Science, {FOCS} 2024}, pages 314--327. {IEEE}, 2024.
\newblock Available at \url{https://arxiv.org/abs/2404.06069}.

\bibitem[BGS20]{BernsteinGS20}
Aaron Bernstein, Maximilian~Probst Gutenberg, and Thatchaphol Saranurak.
\newblock Deterministic decremental reachability, {SCC}, and shortest paths via directed expanders and congestion balancing.
\newblock In {\em Proc. 61st {IEEE} Ann. Symp. Foundations of Comp. Sci.}, pages 1123--1134, 2020.
\newblock Available at \url{https://arxiv.org/abs/2009.02584}.

\bibitem[BK23]{BlikstadK23}
Joakim Blikstad and Peter Kiss.
\newblock Incremental $(1-\epsilon)$-approximate dynamic matching in ${O}(\poly(1/\epsilon))$ update time.
\newblock In {\em Proc. 31st Ann. European Symp. Algorithms}, volume 274 of {\em LIPIcs}, pages 22:1--22:19, 2023.
\newblock Available at \url{https://arxiv.org/abs/2302.08432}.

\bibitem[BKS23]{BhattacharyaKS23dynamic1}
Sayan Bhattacharya, Peter Kiss, and Thatchaphol Saranurak.
\newblock Dynamic $(1+\epsilon)$-approximate matching size in truly sublinear update time.
\newblock In {\em Proc. 64th {IEEE} Ann. Symp. Foundations of Computer Science}, 2023.
\newblock Available at \url{https://arxiv.org/abs/2302.05030}.

\bibitem[CST25]{ChenST23}
Jiale Chen, Aaron Sidford, and Ta{-}Wei Tu.
\newblock Entropy regularization and faster decremental matching in general graphs.
\newblock In {\em Proceedings of the 2025 Annual {ACM-SIAM} Symposium on Discrete Algorithms, {SODA} 2025}, pages 3069--3115. {SIAM}, 2025.
\newblock Available at \url{https://arxiv.org/abs/2312.09077}.

\bibitem[Dah16]{Dahlgaard16}
S{\o}ren Dahlgaard.
\newblock On the hardness of partially dynamic graph problems and connections to diameter.
\newblock In {\em 43rd International Colloquium on Automata, Languages, and Programming, {ICALP} 2016}, volume~55 of {\em LIPIcs}, pages 48:1--48:14. Schloss Dagstuhl - Leibniz-Zentrum f{\"{u}}r Informatik, 2016.
\newblock Available at \url{https://arxiv.org/abs/1602.06705}.

\bibitem[DP14]{DuanP14}
Ran Duan and Seth Pettie.
\newblock Linear-time approximation for maximum weight matching.
\newblock {\em J. {ACM}}, 61(1):1:1--1:23, 2014.

\bibitem[FMU22]{FischerMU22}
Manuela Fischer, Slobodan Mitrovic, and Jara Uitto.
\newblock Deterministic (1+\emph{{\(\epsilon\)}})-approximate maximum matching with poly(1/\emph{{\(\epsilon\)}}) passes in the semi-streaming model and beyond.
\newblock In {\em {STOC} '22: 54th Annual {ACM} {SIGACT} Symposium on Theory of Computing}, pages 248--260. {ACM}, 2022.
\newblock Available at \url{https://arxiv.org/abs/2106.04179}.

\bibitem[GKMS19]{GamlathKMS19}
Buddhima Gamlath, Sagar Kale, Slobodan Mitrovic, and Ola Svensson.
\newblock Weighted matchings via unweighted augmentations.
\newblock In {\em Proc. 2019 {ACM} Symp. Principles of Distributed Computing}, pages 491--500. {ACM}, 2019.
\newblock Available at \url{https://arxiv.org/abs/1811.02760}.

\bibitem[GLS{\etalchar{+}}19]{GrandoniLSSS19}
Fabrizio Grandoni, Stefano Leonardi, Piotr Sankowski, Chris Schwiegelshohn, and Shay Solomon.
\newblock {(1} + {\(\epsilon\)})-approximate incremental matching in constant deterministic amortized time.
\newblock In {\em Proc. 30th Ann. {ACM-SIAM} Symp. Discrete Algorithms}, pages 1886--1898. {SIAM}, 2019.

\bibitem[GP13]{GuptaP13}
Manoj Gupta and Richard Peng.
\newblock Fully dynamic $(1+\epsilon)$-approximate matchings.
\newblock In {\em Proc. 54th Ann. {IEEE} Symp. Foundations of Comp. Sci.}, pages 548--557, 2013.
\newblock Available at \url{https://arxiv.org/abs/1304.0378}.

\bibitem[Gup14]{Gupta14}
Manoj Gupta.
\newblock Maintaining approximate maximum matching in an incremental bipartite graph in polylogarithmic update time.
\newblock In {\em Proc. 34th Int. Conference on Foundation of Software Technology and Theor. Comp, Sci.}, volume~29 of {\em LIPIcs}, pages 227--239, 2014.

\bibitem[HKNS15]{HenzingerKNS15}
Monika Henzinger, Sebastian Krinninger, Danupon Nanongkai, and Thatchaphol Saranurak.
\newblock Unifying and strengthening hardness for dynamic problems via the online matrix-vector multiplication conjecture.
\newblock In {\em Proceedings of the 47th Annual {ACM} Symposium on Theory of Computing, {STOC} 2015}, pages 21--30, 2015.
\newblock Available at \url{https://arxiv.org/abs/1511.06773}.

\bibitem[HS23]{HuangS23}
Shang{-}En Huang and Hsin{-}Hao Su.
\newblock $(1-\eps)$-approximate maximum weighted matching in $\poly(1/\eps,\log n)$ time in the distributed and parallel settings.
\newblock In {\em Proceedings of the 2023 {ACM} Symposium on Principles of Distributed Computing, {PODC} 2023}, pages 44--54. {ACM}, 2023.
\newblock Available at \url{https://arxiv.org/abs/2212.14425}.

\bibitem[JJST22]{JambulapatiJST22}
Arun Jambulapati, Yujia Jin, Aaron Sidford, and Kevin Tian.
\newblock Regularized box-simplex games and dynamic decremental bipartite matching.
\newblock In {\em Proc. 49th Int. Colloquium on Automata, Languages, and Programming}, volume 229 of {\em LIPIcs}, pages 77:1--77:20, 2022.
\newblock Available at \url{https://arxiv.org/abs/2204.12721}.

\bibitem[Liu24]{Liu24}
Yang~P. Liu.
\newblock On approximate fully-dynamic matching and online matrix-vector multiplication.
\newblock In {\em 65th {IEEE} Annual Symposium on Foundations of Computer Science, {FOCS} 2024}, pages 228--243. {IEEE}, 2024.
\newblock Available at \url{https://arxiv.org/abs/2403.02582}.

\bibitem[McG05]{McGregor05}
Andrew McGregor.
\newblock Finding graph matchings in data streams.
\newblock In {\em Approximation, Randomization and Combinatorial Optimization, Algorithms and Techniques, 8th International Workshop on Approximation Algorithms for Combinatorial Optimization Problems, {APPROX} 2005 and 9th InternationalWorkshop on Randomization and Computation, {RANDOM} 2005}, volume 3624 of {\em Lecture Notes in Computer Science}, pages 170--181. Springer, 2005.

\bibitem[MMSS25]{MitrovicMSS25}
Slobodan Mitrovic, Anish Mukherjee, Piotr Sankowski, and Wen{-}Horng Sheu.
\newblock Faster semi-streaming matchings via alternating trees.
\newblock In {\em 52nd International Colloquium on Automata, Languages, and Programming, {ICALP} 2025}, volume 334 of {\em LIPIcs}, pages 119:1--119:19. Schloss Dagstuhl - Leibniz-Zentrum f{\"{u}}r Informatik, 2025.
\newblock Available at \url{https://arxiv.org/abs/2412.19057}.

\bibitem[MS25]{MitrovicS25}
Slobodan Mitrovic and Wen{-}Horng Sheu.
\newblock A framework for boosting matching approximation: parallel, distributed, and dynamic.
\newblock In {\em Proceedings of the 37th {ACM} Symposium on Parallelism in Algorithms and Architectures, {SPAA} 2025, Portland, OR, USA, 28 July 2025 - 1 August 2025}, pages 443--457. {ACM}, 2025.
\newblock Available at \url{https://arxiv.org/abs/2503.01147}.

\bibitem[Pet12]{Pettie12}
Seth Pettie.
\newblock A simple reduction from maximum weight matching to maximum cardinality matching.
\newblock {\em Inf. Process. Lett.}, 112(23):893--898, 2012.

\bibitem[Pra25]{Pratt25}
Kevin Pratt.
\newblock A note on ordered ruzsa-szemer{\'{e}}di graphs.
\newblock In {\em arXiv Preprint}, 2025.
\newblock Available at \url{http://arxiv.org/abs/2502.02455}.

\bibitem[Sch03]{schrijver2003combinatorial}
Alexander Schrijver.
\newblock {\em Combinatorial optimization: polyhedra and efficiency}, volume~24.
\newblock Springer, 2003.

\bibitem[SW17]{StubbsW17}
Daniel Stubbs and Virginia~Vassilevska Williams.
\newblock Metatheorems for dynamic weighted matching.
\newblock In {\em Proc. 8th Innovations in Theor. Comp. Sci. Conference}, volume~67 of {\em LIPIcs}, pages 58:1--58:14, 2017.

\end{thebibliography}
\appendix
\section{Proof of~\cref{lemma:add-apx-to-multi-apx}}\label{appendix:vertex-sparsification}
We generalize the vertex sparsification ideas in~\cite{BhattacharyaKS23dynamic1} to weighted matching problems.

\begin{lemma}[Implied by~{\cite[Lemma 7.2]{BhattacharyaKS23dynamic1}}]\label{lemma:vertex sparsification:unweighted}
There exists a dynamic algorithm $\A$ with $O(\log^2n\cdot \eps^{-2})$ worst-case update time, which maintains: a fixed set of $K=O(\log^2n\cdot \eps^{-2})$ vertex partitions $\{\Gamma_1,\dots,\Gamma_K\}$ and the corresponding contracted graphs $\{G/{\Gamma_1},\dots,G/{\Gamma_K}\}$ of an $n$-node graph $G$ undergoing edge insertions and deletions, and a dynamically changing subset $I\subseteq [1,K]$. Throughout the sequence of updates (w.h.p.\ against an adaptive adversary) the algorithm ensures that:
\begin{enumerate}
    \item $|V/{\Gamma_i}|=\Theta\left(\frac{\mu(G)}{\eps}\right)$ for all $i\in I$;
    \item for any subset $S\subseteq V$ of size $2\cdot \mu(G)$, there is an index $i^*\in I$ such that there are at least $(1-\eps)\cdot |S|$ vertex subsets in $\Gamma_{i^*}$ that contain exactly one vertex in $S$.
\end{enumerate}
\end{lemma}

\begin{corollary}[Weighted Extension of {\cite[Lemma 7.2]{BhattacharyaKS23dynamic1}}]\label{lemma:vertex sparsification:weighted}
There exists a dynamic algorithm $\A$ with $O(\log^2n\cdot \eps^{-2})$ worst-case update time, which maintains: a fixed set of $K=O(\log^2n\cdot \eps^{-2})$ vertex partitions $\{\Gamma_1,\dots,\Gamma_K\}$ and the corresponding contracted graphs $\{G/{\Gamma_1},\dots,G/{\Gamma_K}\}$ of an $n$-node $W$-aspect ratio graph $G$ undergoing edge insertions and deletions, and a dynamically changing subset $I\subseteq [1,K]$. Throughout the sequence of updates (w.h.p.\ against an adaptive adversary) the algorithm ensures that:
\begin{enumerate}
    \item $|V/{\Gamma_i}|=\Theta\left(\frac{\mu(G)}{\eps}\right)$ for all $i\in I$;
    \item there is an index $i^*\in I$ such that $(1-\eps\cdot W)\cdot \mu_{\bw}(G)\leq \mu_{\bw}(G/\Gamma_{i^*})\leq \mu_{\bw}(G)$
\end{enumerate}
\end{corollary}

\begin{proof}
    Consider an arbitrary MWM $M$ and any vertex subset $S$ with size $2\cdot \mu(G)$ that contains all the endpoints of $M$. By~\cref{lemma:vertex sparsification:unweighted}, there is an index $i^*\in I$ s.t.\ there are at least $(1-\eps)\cdot 2\cdot \mu(G)$ vertex subsets that contain exactly one vertex in $S$. Among them, at least $(1-\eps)\cdot 2\cdot \mu(G)-2\cdot (\mu(G)-|M|)=2|M|-2\cdot \eps\cdot \mu(G)$ are the endpoints of $M$. Thus at least $|M|-\eps\cdot \mu(G)$ edges in $M$ are kept in the contracted graph, and $\mu_{\bw}(G/\Gamma_{i^*})\geq \bw(M)-\eps\cdot \mu(G)\cdot W\geq (1-\eps\cdot W)\cdot \mu_{\bw}(G)$.
\end{proof}

\begin{lemma}\label{lemma:add-apx-to-multi-apx:weak}
    Given a $(1,\eps\cdot W\cdot n)$-approximate dynamic MWM algorithm $\A$ that, on an input $n$-vertex $m$-edge $W$-aspect ratio graph, has initialization time $\I(n,m,W,\eps)$ and amortize/worst-case update time $\U(n,m,W,\eps)$, there is a $(1-O(\eps),0)$-approximate dynamic MWM algorithm (w.h.p.\ against an adaptive adversary), on an input $n$-vertex $m$-edge $W$-aspect ratio graph, that has initialization time
    \[O(\log^2 n\cdot \eps^{-2}W^2)\cdot \left(\I(n,m,W,\eps^2W^{-2})+n\right),\]
    and amortize/worst-case update time
    \[O(\log^2 n\cdot \eps^{-2}W^2)\cdot \U(n,m,W,\eps^2W^{-2}).\]
    If $\A$ is fully dynamic/incremental/decremental/offline dynamic, then so is the new algorithm.
\end{lemma}

\begin{proof}
We apply~\cref{lemma:vertex sparsification:weighted} with parameter $\eps^\prime=\eps/W$. Specifically, we use $O(\log^2 n\cdot \varepsilon^{\prime-2})$ copies of $\A$ to maintain a $(1,(\eps\cdot\eps^{\prime}\cdot W^{-1})\cdot W\cdot |V/\Gamma_i|)$-approximate MWM for each contracted graph $G/\Gamma_i$. The weight of the matching we maintain on $G/\Gamma_{i^*}$ is at least
\[\mu_{\bw}(G/\Gamma_{i^*})-\eps\cdot \eps^{\prime}\cdot |V/\Gamma_{i^*}|\geq (1-\eps^\prime\cdot W)\cdot \mu_{\bw}(G)-\eps\cdot\eps^\prime\cdot O\left(\frac{\mu(G)}{\eps^\prime}\right)=(1-O(\eps))\cdot \mu_{\bw}(G),\]
where the first inequality comes from~\cref{lemma:vertex sparsification:weighted} and the second equality comes from the choice of $\eps^\prime$. The initialization of the algorithm includes the setup for~\cref{lemma:vertex sparsification:weighted} which takes $O(\log^2 n\cdot \eps^{\prime -2})\cdot n$ time and the initilization of $O(\log^2 n\cdot \eps^{\prime-2})$ copies of $\A$. The update time comes from the $O(\log^2 n\cdot \eps^{\prime-2})$ update time in~\cref{lemma:vertex sparsification:weighted} and the update time of $O(\log^2 n\cdot \eps^{\prime-2})$ copies of $\A$.
\end{proof}

We then use the following weight reduction framework by~\cite{BernsteinCDLST24} to reduce $W$ to $\eps^{-2}$.

\begin{lemma}[Weight Reduction,~{\cite[Theorem 3.5]{BernsteinCDLST24}}]\label{lemma:weight-reduction}
Given a dynamic $(1-\eps,0)$-approximate MWM algorithm $\A$ that, on input $n$-node $m$-edge graph with aspect ratio $W$, has initialization time $\I(n,m,W,\eps)$, and update time $\U(n,m,W,\eps)$,
there is a transformation which produces a dynamic $(1-\eps\log(\eps^{-1}),0)$-approximate MWM algorithm that has an initialization time
\[\log(\eps^{-1})\cdot O(\I(n,m,\Theta(\eps^{-2}),\Theta(\eps)) + m\eps^{-1}),\]
amortized update time
\[\poly(\log(\eps^{-1}))\cdot O(\U(n,m,\Theta(\eps^{-2}),\Theta(\eps)) + \eps^{-5}).\]
If $\A$ is fully dynamic/incremental/decremental/offline dynamic, then so is the new algorithm.
\end{lemma}

\cref{lemma:add-apx-to-multi-apx} is a corollary of~\cref{lemma:add-apx-to-multi-apx:weak,lemma:weight-reduction}.

\section{Algorithm for~\cref{lemma:structural}}\label{sec:structural-lemma-proof}

We will use~\cite{MitrovicMSS25} to augment the current matching until the case where after removing a relatively small number of edges, the structures maintained in the algorithm can be easily modified to generate the output for~\cref{lemma:structural}. In~\cref{appendix:high-level}, we will give a high-level description of~\cite{MitrovicMSS25}. The original version of~\cite{MitrovicMSS25} is a streaming algorithm; in~\cref{appendix:our-implementation}, we implement the algorithm with approximate induced matching oracles instead of edge streams. In~\cref{appendix:properties}, we prove several properties of the algorithm that are useful for output construction. Finally, in~\cref{appendix:output-construction,appendix:proof-of-structural-lemma}, we construct the output for~\cref{lemma:structural} and prove the properties and the query complexity.

\subsection{A High-level Description of~\cite{MitrovicMSS25}}\label{appendix:high-level}
On a high level (\cref{alg:MMSS24:high-level}), the MMSS algorithm runs on exponential scales and augments the current matching for multiple phases with \texttt{ALG-PHASE} (\cref{alg:MMSS24:AlgPhase}).
\begin{algorithm}[!ht]
    \caption{\cite[Algorithm 1]{MitrovicMSS25}} \label{alg:MMSS24:high-level}
    
    \SetEndCharOfAlgoLine{}
    \SetKwInput{KwData}{Input}
    \SetKwInput{KwResult}{Output}
    \SetKwFunction{AlgPhase}{ALG-PHASE}
    \KwData{A graph $G$ and the approximation parameter $\eps$}
    \KwResult{A $(1-\eps,0)$-approximate MCM.}
    Compute a $\frac{1}{2}$-approximate MCM $M$.\label{line:half-apx}\;
    \For{scale $h=\frac{1}{2},\frac{1}{4},\frac{1}{8},\dots,\frac{\eps^2}{64}$} {
        \For{phases $t=1,2,3,\dots,\frac{144}{h\eps}$} {
            $\P\gets$\AlgPhase{$G,M,\eps,h$}.\;
            Restore all vertices removed in the execution of \AlgPhase.\;
            Augment the current matching $M$ using the vertex-disjoint augmenting paths in $\P$.
        }
    }
    \textbf{return} $M$
\end{algorithm}

The algorithm for a single phase, \AlgPhase, takes a graph $G$, a matching $M$, a parameter $\eps$, and a scale $h$, searching for a set of vertex-disjoint augmenting paths with length at most $\Theta(1/\eps)$. The smaller the scale $h$, the more time would be spent in \AlgPhase to search for the augmenting paths. 
\begin{algorithm}[!ht]
    \caption{\texttt{ALG-PHASE},~\cite[Algorithm 2]{MitrovicMSS25}} \label{alg:MMSS24:AlgPhase}
    
    \SetEndCharOfAlgoLine{}
    \SetKwInput{KwData}{Input}
    \SetKwInput{KwResult}{Output}
    \SetKwFunction{Extend}{EXTEND-ACTIVE-PATH}
    \SetKwFunction{ContractAndAugment}{CONTRACT-AND-AUGMENT}
    \SetKwFunction{Backtrack}{BACKTRACK-STUCK-STRUCTURES}
    \SetKwFunction{OurExtend}{OUR-EXTEND-ACTIVE-PATH}
    \SetKwFunction{OurContractAndAugment}{OUR-CONTRACT-AND-AUGMENT}
    \KwData{A graph $G$, the current matching $M$, the approximation parameter $\eps$ and the scale $h$}
    \KwResult{A set $\P$ of vertex-disjoint augmenting paths w.r.t.\ $M$.}
    $\P\gets \emptyset,\texttt{limit}_h=\frac{6}{h}+1,\tau_{\max}(h)=\frac{72}{h\eps}$.\;
    $\ell(a)\gets \ell_{\max}+1$ for each arc $a\in M$.\;
    Initialize the structure $S_\alpha$ for each free vertex $\alpha$.\;
    \For{$\tau=1,2,\dots,\tau_{\max}(h)$\label{line:streaming pass bundle}}{
        \For{each free vertex $\alpha$}{
            \eIf{$S_\alpha$ has at least $\emph{\texttt{limit}}_h$ vertices}{mark $S_\alpha$ as ``on hold''.}{mark $S_\alpha$ as ``not on hold''.}
            Mark $S_\alpha$ as ``not modified''.\;
        }
        \Extend\;
        \ContractAndAugment\;
        \Backtrack\;
    }
    \textbf{return} $\P$.\;
\end{algorithm}

The searches are carried out by growing alternating trees from all free vertices $\alpha$ simultaneously w.r.t. a laminar family of blossom $\Omega$.
\begin{definition}[Alternating tree, inner and outer vertex]
    An alternating tree is a rooted tree with a free vertex as the root and any path from the root to a leaf is an even-length alternating path. For any vertex within the tree, if the number of edges in the alternating path from the root to it is odd, it is called an inner vertex, otherwise, it is called an outer vertex.
\end{definition}
The algorithm maintains structures $S_\alpha$ for the free vertices $\alpha$:

\begin{definition}[{\cite[Definition 4.1]{MitrovicMSS25}}]
    A structure $S_\alpha=(G_\alpha,w^\prime_\alpha)$ of a free node $\alpha$ is a pair of a subgraph $G_\alpha$ of $G$ such that $G_{\alpha}/\Omega$ is an alternating tree, and a working vertex $w^\prime_\alpha$ in $G_\alpha$ which could be $\emptyset$ with the following properties:
    \begin{enumerate}
        \item \textbf{Disjointness:} The subgraphs $G_\alpha$ are vertex-disjoint for different free vertices $\alpha$.
        \item \textbf{Tree representation:} If the subgraph $G_\alpha$ contains an arc $(u,v)$ with $\Omega(u)\neq \Omega(v)$ then $\Omega(u)$ is the parent of $\Omega(v)$ in the alternating tree $G_\alpha/\Omega$.
        \item \textbf{Unique arc property:} For each arc $(u^\prime,v^\prime)\in G_\alpha/\Omega$, there is a unique arc $(u,v)\in G_\alpha$ such that $\Omega(u)=u^\prime$ and $\Omega(v)=v^\prime$.
    \end{enumerate}
    \begin{remark}
        In the algorithm, during the initialization of the structures, $\Omega$ will be set to the union of all trivial blossoms. We will also use $\Omega(u)$ for any vertex $u\in G$ to represent the root blossom in $\Omega$ that contains $u$.
    \end{remark}
\end{definition}

    The working vertex of the structure is the current visiting vertex in the DFS, and the active path is the alternating path between the working vertex and the root. A vertex or an arc is called active if and only if it is on an active path. In addition, each matched arc $a$ has a label $\ell(a)$, which will be updated non-increasingly during the search.

    \AlgPhase grows the structures with restrictions on the search depth and the size limit of each structure so that each alternating tree does not grow too deep or too large. More specifically, there is a depth limit $\ell_{\max}=\frac{3}{\eps}$ and a size limit $\texttt{limit}_{h}=\frac{6}{h}$. The structure with size at least $\texttt{limit}_h$ will be put ``on hold'' so it won't grow in the next round. Those parameters and the number of rounds are chosen so that if we start the phase with a matching $M$ at a scale $h$, at the end of the phase there will be at most $h\cdot |M|$ active free vertices (\cref{lemma:number-of-active-structures}). 
    \begin{algorithm}[!ht]
    \caption{\texttt{EXTEND-ACTIVE-PATH},~\cite[Algorithm 3]{MitrovicMSS25}} \label{alg:MMSS24:Extend}
    
    \SetKwFunction{Contract}{CONTRACT}
    \SetKwFunction{Augment}{AUGMENT}
    \SetKwFunction{Overtake}{OVERTAKE}
    \SetEndCharOfAlgoLine{}
    \For{each arc $g=(u,v)\in E(G)$}{
        \lIf{$u$ or $v$ was removed in this phase, or $g$ is matched}{\textbf{continue}.}
        \lIf{$\Omega(u)$ is not a working vertex, or $\Omega(u)=\Omega(v)$}{\textbf{continue}.}
        \lIf{$u$ belongs to a structure ``on hold'' or marked modified}{\textbf{continue}.}
        \eIf{$\Omega(v)$ is an outer vertex}{
            \eIf{$\Omega(u)$ and $\Omega(v)$ are in the same structure}{
                \Contract{$g$}.
            }{
                \Augment{$g$}.
            }
        }{
            Compute $d(u)$.\;
            $a\gets$ the matched arc in $G$ whose tail is $v$.\;
            \lIf{$d(u)+1<\ell(a)$}{
                \Overtake{$g,a,d(u)+1$}.
            }
        }
    }
\end{algorithm}
    
    The execution of \AlgPhase contains three subroutines. \Extend (\cref{alg:MMSS24:Extend}) tries to grow the structure $S_\alpha$ for each free vertex $\alpha$. It attempts to visit a vertex $v$ (with an associate matched arc $a$) neighboring to its working vertices through an unmatched arc $(u,v)$, i.e., $\Omega(u)=w^\prime_\alpha$ and $v\notin \Omega(u)$. If $\Omega(v)$ is an outer vertex, one of the two things could happen: (1) $\Omega(u)$ and $\Omega(v)$ belong to the same structure, or (2) they don't. In case (1) we find a new blossom, the algorithm invokes \Contract which contracts the blossom and updates the structure; In case (2) we find an augmenting path, the algorithm invokes \Augment which adds the corresponding augmenting path into $\P$ and removes all vertices in both structures from the graph. Note that these vertices will be added back before the next \AlgPhase is called. If, on the other hand, $\Omega(v)$ is not an outer vertex, meaning it is unvisited or an inner vertex, the algorithm looks at $\ell(a)$. If $\ell(a)$ is larger than $d(u)+1$, it invokes \Overtake, grows the alternating tree $S_\alpha$ and reduces the label of $\ell(a)$ to $d(u)+1$. Here $d(u)$ is 0 if $\Omega(u)$ is free, or $\ell(a^\prime)$ for corresponding matched arc $a^\prime$ associated with $\Omega(u)$. Though ``on hold'' structures are not allowed to grow, they could still be overtaken and thus the algorithm could produce a structure of size at most $\texttt{limit}_h\cdot \ell_{\max}$ (\cref{lemma:structure size}). We refer the interested readers to{~\cite[section 4.5]{MitrovicMSS25}} for the details of \Augment,\Contract and \Overtake.

    In \ContractAndAugment, the algorithm looks at every arc $a=(u,v)$ such that $\Omega(u)$ is a working vertex, both $\Omega(u)$ and $\Omega(v)$ are outer vertices and $\Omega(u)\neq \Omega(v)$. If $u$ and $v$ are in the same structure, it invokes \Contract; otherwise, it invokes \Augment. 

    In \Backtrack, the algorithm looks at structures that fail to progress in \Extend or \ContractAndAugment, i.e., those not ``on hold'' and not modified, and updates the working vertices of those structures to be their parents. If the working vertex was already the root, set the new one to be $\emptyset$.

\subsection{Our Implementation with $(\alpha,\beta)$-Approximate Induced Matching Oracle}\label{appendix:our-implementation}

We now show how to use the approximate induced matching oracle on $\B_G$ to implement~\cite{MitrovicMSS25}. We highlight the changes made to different parts of the algorithm.

Line~\ref{line:half-apx} of~\cref{alg:MMSS24:high-level} finds a $\frac{1}{2}$-approximate matching in the graph, but we will start with a given matching. Additionally, instead of having multiple scales, we fix a scale $h$ and run the algorithm until the number of removed structures is at most $h\eps^2 n$ within the last phase.

\cref{alg:MMSS24:AlgPhase} stays the same except that we substitute the \Extend and\\ \ContractAndAugment procedures with \OurExtend and \OurContractAndAugment. 

For \OurExtend, instead of scanning over all arcs in the graph, we use the induced matching oracle to find arcs that would lead to an extension (\cref{alg:ours:Extend}). More specifically, we denote $W(i)$ as all $u$ in $G$ with $d(u)=i$ such that $\Omega(u)$ is a working vertex in $G/\Omega$ that belongs to a structure not removed, not modified, and marked ``not on hold'', $I(j)$ as all valid non-outer vertices in $G$ (including inner vertices in structures not removed and vertices not belonging to any structure) whose matched arc has label $j$. The arcs $(u,v)$ that could lead to an extension must have the following form: $u\in W(i)$ and $v\in I(j)$ for $j>i+1$ (recall that matched arcs not belonging to any structure have label $\ell_{\max}+1$). Using induced matching oracles for each $i$ one by one, we can find such arcs and extend the structures according to them.

\begin{algorithm}[!ht]
    \SetKwInput{KwData}{Input}
    \caption{\texttt{OUR-EXTEND-ACTIVE-PATH}} \label{alg:ours:Extend}
    \KwData{An $(\alpha,\beta)$-approximate induced matching oracle $\O$ on $\B_G$.}
    \SetEndCharOfAlgoLine{}
    \For{$i=0,1,\dots,\ell_{\max}-1$}{
        Set $S^1_i\gets\{u^1:u\in W(i)\}$ and $S^2_i\gets\{v^2:v\in\bigcup_{j>i+1} I(j)\}$.
    }
    \While{$\exists~i,~|\O(S^1_i\cup S^2_i)|\geq \beta\cdot n$\label{line:our-extend-unexplored-edges}}{
        \For{any edge $\{u^1,v^2\}\in \O(S^1_i\cup S^2_i)$}{
            $a\gets$ the matched arc in $G$ whose tail is $v$.\;
            \Overtake{$(u,v),a,i+1$}.\;
        }
        Update $S^1_j$ and $S^2_j$ for all $j$ according to the changes in $W(i)$ and $\bigcup_{j>i+1}I(j)$.
    }
\end{algorithm}
A similar change happens in \OurContractAndAugment (\cref{alg:ours:augment}). We use an induced matching oracle on $\B_{G/\Omega}$ to find augmenting paths, which can be built upon an induced matching oracle on $\B_G$ according to~\cref{lemma:reduction:contracted-to-uncontracted}.

\begin{algorithm}[!ht]
    \SetKwInput{KwData}{Input}
    \caption{\texttt{OUR-CONTRACT-AND-AUGMENT}} \label{alg:ours:augment}
    \SetEndCharOfAlgoLine{}
    \KwData{An $(\alpha,\beta)$-approximate induced matching oracle $\O$ on $\B_G$, $\gamma=\max_{B\in \Omega} |B|$.}
    Build an $(\Omega(\alpha/\gamma^2),O(\beta\gamma))$-approximate induced matching oracle $\O^\prime$ on $\B_{G/\Omega}$ based on~\cref{lemma:reduction:contracted-to-uncontracted}.\;
    \For{all structure $S_\alpha$}{
        \While{$\exists$ an arc $g=(u,v)$ s.t.\ $\Omega(u)=w^\prime_{\alpha}$ and $\Omega(v)\neq w^\prime_{\alpha}$ is an outer vertex in $S_{\alpha}$}{
            \Contract{$g$}.
        }
    }
    Consider any vertex $u\in V$ s.t.\ it is not removed and $\Omega(u)$ is an outer vertex, then set $S^1$ as the set of $(\Omega(u))^1$ and $S_2$ as the set of $(\Omega(u))^2$.\;
    \While{$|\O^\prime(S^1\cup S^2)|\geq \beta\cdot \gamma\cdot n$\label{line:our-contract-and-augment-unexplored-edges}}{
        \For{any edge $\{B^1,B^2\}\in\O^\prime(S^1\cup S^2)$}{
            Pick any edge $\{u,v\}$ s.t.\ $(\Omega(u))^1=B^1$ and $(\Omega(v))^2=B^2$.\;
            \Augment{$(u,v)$}.\;
        }
        Update $S^1$ and $S^2$ according to the removed structures.
    }
\end{algorithm}

\begin{remark}
Before \Backtrack, a structure marked ``not modified'' may not finish exploring. Some adjacent edges are left unexplored because of the additive error in the approximate induced matching oracle. However, we will show that the set of left edges is small in terms of its matching size.
\end{remark}

\subsection{Key Properties of the Algorithm}\label{appendix:properties}
After the while loops at line~\ref{line:our-extend-unexplored-edges} in~\cref{alg:ours:Extend} and line~\ref{line:our-contract-and-augment-unexplored-edges} in~\cref{alg:ours:augment} end, since we use an approximate induced matching oracle, some arcs are left undiscovered. We denote $E_\tau$ as the union of the corresponding undirected edges in the first $\tau-1$ rounds.
Denote $\Omega^\tau$ as the set of blossoms, $\ell^\tau$ as the labels of arcs, and $d^\tau$ as the $d$-function at the beginning of the $\tau$-th round.

\begin{observation}[Similar to~{\cite[Lemma 5.3]{MitrovicMSS25}}]\label{lemma:working-vertex}
    Suppose $G/\Omega^\tau$ contains an inactive outer vertex $v$ not removed at the beginning of the $\tau$-th round, there exists $\tau^\prime\leq \tau$ such that $v$ is a working vertex at the beginning of the $\tau^\prime$-th round.
\end{observation}

\begin{proof}
Suppose at the $\tau^{\prime\prime}$-th round where $\tau^{\prime\prime}<\tau$, $v$ is first added to a structure. 
Right after $v$ is added to a structure, it is the working vertex of that structure by the definition of \Overtake or \Contract. $v$ may be overtaken by other structures in the same round, but \Overtake would again set $v$ as the working vertex of the new structure by definition. And since $v$ is an outer vertex in $G/\Omega^\tau$ at the beginning of the $\tau$-th round and not removed, there is no \Contract or \Augment happened to $v$ after its formation in the $\tau^{\prime\prime}$-th round. Also, \Backtrack does not backtrack $v$'s structure because it is marked modified by either \Overtake or \Contract that adds $v$ into the structure. Therefore, $v$ is a working vertex at the beginning of $\tau^{\prime\prime}+1$ round.
\end{proof}

\begin{observation}[Similar to~{\cite[Invariant 5.4]{MitrovicMSS25}}]\label{lemma:no-outer-outer-edges}
    At the beginning of the $\tau$-th round, any edge connecting two outer vertices that are not removed must be included in $E_{\tau}$.
\end{observation}

\begin{proof}
    Suppose, that there is an edge $\{u,v\}$ where $\Omega^\tau(u)\neq \Omega^\tau(v)$ are both outer vertices and not removed. Denote $\tau_u<\tau$ (resp.\ $\tau_v<\tau$) as the first round in which $\Omega^\tau(u)$ (resp.\ $\Omega^\tau(v)$) is added to a structure. Assume w.l.o.g.\ that $\tau_u\geq \tau_v$. The proof of~\cref{lemma:working-vertex} shows that $\Omega^\tau(u)$ is a working vertex after its formation in the $\tau_u$-th round. $\Omega^\tau(u)$ and $\Omega^\tau(v)$ do not belong to the same structure before \OurContractAndAugment at the $\tau_u$-th round, otherwise they will be included into a larger blossom during \OurContractAndAugment. Since $\Omega^\tau(u)$ and $\Omega^\tau(v)$ are not removed, we know that $(\Omega^\tau(u)^1,\Omega^\tau(v)^2)$ is not output by the approximate induced matching oracle in line~\ref{line:our-contract-and-augment-unexplored-edges} of \OurContractAndAugment before the output size is too small, thus $\{u,v\}$ is included in $E_\tau$.
\end{proof}

\begin{observation}\label{lemma:blossom-inner-outer}
    Consider a vertex $v\in G$. If $\Omega^{\tau}(v)$ is not an outer vertex, it only contains $v$. If $\Omega^\tau(v)$ is an outer vertex, for any $\tau^\prime\geq \tau$, $\Omega^{\tau^\prime}(v)$ is an outer vertex.
\end{observation}

\begin{proof}
    By definition of \Contract and \Overtake.
\end{proof}

\begin{lemma}\label{lemma:no-extending-unmatched-arcs}
    At the beginning of the $\tau$-th round, for an unmatched arc $(u^\prime,v^\prime)$ and a matched arc $a=(u,v)$ such that $u^\prime,v^\prime,u,v$ are not removed, $\Omega^\tau(u^\prime)$ is an outer vertex, $d^{\tau}(u^\prime)<\ell^\tau(a)-1$ and $\Omega^\tau(v^\prime)=\Omega^\tau(u)$, then either $\Omega^\tau(u^\prime)$ is active or $\{u^\prime,v^\prime\}\in E_\tau$.
\end{lemma}

\begin{proof}
    Suppose $\Omega^\tau(u^\prime)$ is not active. Consider the last round $\tau^\prime$ s.t.\ $\Omega^{\tau}(u^\prime)$ is a working vertex at the beginning of that round. By~\cref{lemma:working-vertex}, $\tau^\prime$ is well-defined. 
    In $\tau^\prime$-th round,\\\Backtrack sets $\Omega^\tau(u^\prime)$ to be inactive, thus the structure is ``not on hold'' and not modified in this round. $d^{\tau^\prime}(u^\prime)=d^\tau(u^\prime)$ is not updated afterward.

    Suppose $\Omega^\tau(v^\prime)$ is an outer vertex, by~\cref{lemma:no-outer-outer-edges}, $\{u^\prime,v^\prime\}\in E_{\tau}$. Thus w.l.o.g.\ we assume that $\Omega^\tau(v^\prime)$ is not an outer vertex thus by~\cref{lemma:blossom-inner-outer}, the maximal blossom in $\Omega$ that $v^\prime$ belongs to has always been $\{v^\prime\}$ and not outer at any time. Consequently, $v^\prime=u$ and $a$ is the matched arc w.r.t.\ $v^\prime$. Now consider the time point when \OurExtend of the $\tau^\prime$-th round ends.  The label of $a$ is non-increasing, thus at least $\ell^{\tau}(a)>d^{\tau}(u^\prime)+1$ then.
    Therefore, $u^\prime\in W(d^{\tau^\prime}(u^\prime))$ and $v^\prime\in I(j)$ for some $j>d^{\tau}(u^\prime)+1=d^{\tau^\prime}(u^\prime)+1$. And the fact that $(u^{\prime 1},v^{\prime 2})$ is not found by the approximate induced matching oracle in~\cref{line:our-extend-unexplored-edges} of \OurExtend in this round indicates that $\{u^\prime,v^\prime\}\in E_{\tau}$.
\end{proof}

\subsection{Construction of the Output}\label{appendix:output-construction}
Denote $M^*,\Omega^*,\ell^*,d^*$ as the matching, the laminar family of the blossoms, the matched arcs' labels, and the $d$-function at the end of the algorithm. The matching $M$ output in~\cref{lemma:structural} will be exactly $M^*$. The vertex partition $\Vin,\Vout,\Vunfound$ will be defined by the following procedure. Let $\Vremoved$ contain all the removed vertices in the last phase, $\Vactive$ contain all the vertices $u$ s.t.\ $\Omega^*(u)$ is active at the end of the last phase, $V_i=\{u\in V\setminus (\Vremoved\cup \Vactive):\Omega^*(u)~\text{is an outer vertex}\land d^{*}(u)=i\}$ and let $V_{i^*}$ be the set with minimum size among $V_{0},V_1,\dots,V_{\ell_{\max}}$, then
\[\Vout=\bigcup_{i=0}^{i^*}V_i,\]
\[\Vin=\{u:u~\text{is matched to some vertex in}~\Vout\},\]
\[\Vunfound=V\setminus (\Vin\cup \Vout).\]
The laminar family of blossoms $\Omega$ output in~\cref{lemma:structural} will contain all the trivial blossoms and those non-trivial blossoms $B\in \Omega^*$ such that $B\subseteq \Vout$. Denote $E_{\tau_{\max}(h)+1}$ as the undiscovered edges at the end of the last phase, $\Edel$ is the union of $E_{\tau_{\max}(h)+1}$, unmatched edges adjacent to $v\in \Vactive$ such that $\Omega^*(v)$ is an inner vertex, unmatched edges adjacent to vertices in $\Vremoved$ and unmatched edge adjacent to vertices in $V_{i^*}$. $\Fdel$ is the free vertices in $\Vactive$.

\subsection{Proof of~\cref{lemma:structural}}\label{appendix:proof-of-structural-lemma}

Among the six properties listed in~\cref{lemma:structural}, the definition of our algorithm proves property 1,3,5, property 4 comes from~\cref{lemma:no-outer-outer-edges} and the construction of the vertex partition. At the end of this subsection, we will prove property 2, the query complexity, and the additional runtime of the algorithm.
We will picking $h=\Theta(\eps^2)$ and $\eps=\Theta(\delta^{1/5})$.

\begin{proof}[Proof of Property 1]
In the construction of $\Omega$, any non-trivial blossom is a subset of $\Vout$.
\end{proof}
\begin{proof}[Proof of Property 2]
    Consider any edge $e=\{u,v\}\in M$ s.t.\ $\Omega(u)\neq \Omega(v)$. In the case where neither of $u$ and $v$ belong to any structure, both $u$ and $v$ are in $\Vunfound$. Otherwise, w.l.o.g.,\ we assume that $(u,v)$ is in a structrure. Thus $\Omega^*(u)$ is an inner vertex and $\Omega^*(v)$ is an outer vertex. By construction of the vertex partition, if $v\in \Vout$, then $u\in \Vin$; otherwise, both of them are in $\Vunfound$.
\end{proof}

\begin{proof}[Proof of Property 4]
    For any free vertex $\alpha$, $\Omega^*(\alpha)$ is an outer vertex. For those $\alpha$ not active, $\alpha\in \Vout$ since $\alpha\in V_{0}$. The remaining ones are exactly $\Fdel$, and belong to $\Vunfound$.
\end{proof}

\begin{proof}[Proof of Property 3]
    Any vertex $u\in \Vout$ satisfies that $u$ is not removed, $\Omega^*(u)$ is an outer vertex, and $\Omega^*(u)$ is not active. By~\cref{lemma:no-outer-outer-edges}, any arc connecting two outer vertices that are not removed is in $E_{\tau_{\max}(h)+1}\subseteq \Edel$. Thus we only need to consider the edge $\{u,v\}$ where $u\in \Vout$, $v\in \Vunfound$ and $\Omega^*(v)$ is not an outer vertex. Thus $\Omega^*(u)\neq \Omega^*(v)$. If $u\in V_{i^*}$ or $\Omega^*(v)$ is an active inner vertex or removed, the edge is also in $\Edel$ by construction. Otherwise, $u\notin V_{i^*}$ and either $\Omega^*(v)$ does not belong to any structure or $v$ is matched to a vertex in $V_j$ with $j>i^*$. In either case, the matched arc $a$ corresponding to $\Omega^*(v)$ has label $\ell^*(a)>i^*$.
    
    We now argue that $(u,v)$ and $a$ satisfies the condition of~\cref{lemma:no-extending-unmatched-arcs}, thus $(u,v)\in E_{\tau_{\max}(h)+1}\subseteq \Edel$. First of all, $(u,v)$ is unmatched because otherwise $v\in \Vin$. Also, no related vertices is removed. Finally, $\Omega^*(u)$ is an outer vertex and $d^{*}(u)<i^*\leq \ell^*(a)-1$ since $u\notin V_{i^*}$.
\end{proof}

\begin{proof}[Analysis of $|\Fdel|$ and $\mu(\Edel)$]
    
Denote $U$ as the set of vertices $v$ in $\Vactive$ such that $\Omega^*(v)$ is an inner vertex.
We bound $\mu(\Edel)$ and $|\Fdel|$ by bounding $\mu(E_{\tau_{\max}(h)+1})$, $|U|$, $|\Vremoved|$ and $|V_{i^*}|$.

\begin{lemma}[Upper bound on structure size,~{\cite[Lemma 6.6]{MitrovicMSS25}}]\label{lemma:structure size}
    For a phase with scale $h$, at any moment of the phase, the size of the structure is at most $\texttt{limit}_h\cdot\ell_{\max}=O((h\eps)^{-1})$. Thus, $\gamma=O((h\eps)^{-1})$.
\end{lemma}

\begin{lemma}
    $\mu(E_{\tau_{\max}(h)+1})=O\left((\beta/\alpha)h^{-4}\eps^{-4}\right)\cdot n$.
\end{lemma}

\begin{proof}
    Consider the edge set $E_{\tau}\setminus E_{\tau-1}$. By definition, edges in this set are left unfound by the induced matching oracle in the $(\tau-1)$-th round. An $(\alpha,\beta)$-approximate induced matching oracle stops returning matchings in our algorithm only if the matching size in the current output is less than $\beta\cdot n$. Thus in~\OurExtend, that the queried induced subgraph $\B_{G}[S]$ satisfies $\mu(\B_{G}[S])\leq O(\beta/\alpha)\cdot n$. And in~\OurContractAndAugment, the queried induced subgraph $\B_{G/\Omega}[S]$ satisfies $\mu(\B_{G}[S])\leq O(\beta\gamma^3/\alpha)\cdot n=O(\beta/(\alpha h^3\eps^3))\cdot n$. Since we have $\tau_{\max}(h)$ rounds, 
    \[\mu(E_{\tau_{\max}(h)+1})=O\left(\tau_{\max}(h)\cdot\beta/(\alpha h^3\eps^3) \right)\cdot n=O\left((\beta/\alpha)\cdot h^{-4}\eps^{-4}\right)\cdot n.\]
\end{proof}

\begin{lemma}[Upper bound on the number of active structures,~{\cite[Lemma 6.3]{MitrovicMSS25}}]\label{lemma:number-of-active-structures}
    For a phase with scale $h$, let $M_0$ be the matching at the beginning of the phase. Then at the end of that phase, there are at most $h\cdot |M_0|$ active structures.
\end{lemma}

\begin{corollary}
    $|\Fdel|\leq O(h)\cdot n$ and $|U|\leq O(h\eps^{-1})\cdot n$.
\end{corollary}

\begin{proof}
    $|\Fdel|$ is bounded by the number of active structures. An active structure contains one active path with at most $O(1/\eps)$ inter-blossom edges. Also, by~\cref{lemma:blossom-inner-outer}, if $\Omega^*(v)$ is an inner vertex, the blossom only contains $v$.
\end{proof}

\begin{lemma}
    $|\Vremoved|\leq O(\eps)\cdot n$.
\end{lemma}

\begin{proof}
\cref{lemma:structure size} and the fact that the algorithm stops when at most $h\eps^2 n$ structures are removed in the last phase.
\end{proof}

\begin{lemma}
    $|V_{i^*}|\leq O(\eps)\cdot n$.
\end{lemma}

\begin{proof}
    For different $i$ and $j$, $V_i$ and $V_j$ are disjoint. Since $|V_{i^*}|$ contains the minimum number of vertices, $|V_{i^*}|\leq O(1/\ell_{\max})\cdot n=O(\eps)\cdot n$.
\end{proof}

\begin{corollary}
    $\mu(\Edel)\leq O\left((\beta/\alpha)h^{-4}\eps^{-4}+h\eps^{-1}+\eps\right)\cdot n$.
\end{corollary}

Since we pick $h=\Theta(\eps^2)$ and $\eps=\Theta((\beta/\alpha)^{1/13})$, we finish the proof.

\end{proof}

\begin{proof}[Runtime Analysis]
Since \Contract, \Augment and \Overtake all runs in linear time, the additional runtime of the algorithm is dominated by the updates to the edge oracles and queries to the adjacency matrix. We first bound the number of phases in our algorithm.

\begin{lemma}\label{lemma:number of phases}
    The number of phases is at most $O((h\eps^2)^{-1})$.
\end{lemma}

\begin{proof}
    Recall that the algorithm stops when at most $h\eps^2 n$ structures are removed. A structure is removed if and only if an augmenting path is found. And the augmenting paths found in the same phase are vertex disjoint. Thus when $k$ structures are removed in a phase, the matching size grows $k/2$. Since the matching size cannot exceed $n$, there are at most $O((h\eps^2)^{-1})$ phases.
\end{proof}

Now we consider the total number of queries we make to the induced matching oracles.
\begin{lemma}\label{lemma:induced-matching-oracle-calls}
    The total number of queries to the induced matching oracles is $O(h^{-6}\eps^{-7}\log n\log(1/\beta)\alpha^{-1})$.
\end{lemma}
\begin{proof}
    During~\OurExtend for each $i\in[O(\eps^{-1})]$, we have at most $O(\log(1/\beta)\cdot \alpha^{-1})$ queries because we find an $(\alpha,\beta)$-approximate matching every time and the matched vertices is removed from further queries for the same $i$ in this round. A similar argument combined with~\cref{alg:reduction:contracted-to-uncontracted} applies to~\OurContractAndAugment, which has $O(\gamma^4\log n\log(1/(\beta\gamma))\alpha^{-1})=O(h^{-4}\eps^{-4}\log n\log (h\eps/\beta)\alpha^{-1})$. Thus the total number of queries is $O(h^{-4}\eps^{-4}\log n\log(1/\beta)\alpha^{-1})$ per round. Since there are $\tau_{\max}(h)=O((h\eps)^{-1})$ rounds in each phase, and $O((h\eps^{2})^{-1})$ phases by~\cref{lemma:number of phases}, the total number of updates is $O(h^{-6}\eps^{-7}\log n\log(1/\beta)\alpha^{-1})$.
\end{proof}
Now we bound the number of queries to the adjacency matrix.
\begin{lemma}
    The total number of queries to the adjacency matrix is $O(h^{-4}\eps^{-5})\cdot n$.
\end{lemma}

\begin{proof}
    The queries are all made in \OurContractAndAugment. Since \OurContractAndAugment is called $\tau_{\max}(h)=O((h\eps)^{-1})$ times per phase, and the structure size is bounded by $O((h\eps)^{-1})$ by~\cref{lemma:structure size}, the number of queries is $O((h\eps)^{-1})\cdot n\cdot O((h\eps)^{-2})=O(h^{-3}\eps^{-3})\cdot n$. Additionally, there will be at most $O(n)$ contractions since each contraction reduces the number of vertices in $G/\Omega$. Every contraction will lead to another $O(h^{-2}\eps^{-2})$ queries thus in total $O(h^{-2}\eps^{-2})\cdot n$ queries. By~\cref{lemma:number of phases}, the total number of queries is $O(h^{-4}\eps^{-5})\cdot n$.
\end{proof}
Since we set $h=\Theta(\eps^2)$ and $\eps=\Theta((\beta/\alpha)^{1/13})$, we prove the runtime of the algorithm.
\end{proof}

\end{document}